\newif\ifaddone \addonetrue

\documentclass[12pt]{article}

\usepackage{mmp}

\usepackage{algorithm}

\usepackage{algpseudocode}
\algtext*{EndWhile}% Remove "end while" text
\algtext*{EndFor}
\algtext*{EndIf}% Remove "end if" text
\algnewcommand\algorithmicinput{\textbf{Input:}}
 \algnewcommand\INPUT{\item[\algorithmicinput]}
 \algnewcommand\algorithmicoutput{\textbf{Output:}}
 \algnewcommand\OUTPUT{\item[\algorithmicoutput]}

\usepackage{verbatim}
\usepackage{xspace}
\usepackage{listings}

\usepackage{hyperref}
\hypersetup{bookmarksnumbered=true,
            colorlinks=true,
            citecolor=MediumBlue,urlcolor=MediumBlue,linkcolor=MediumBlue}

\newtheorem{remark}[theorem]{Remark}
\newtheorem*{assumpt}{Assumption}
\newcommand{\cA}{\mathcal{A}}
\newcommand{\cB}{\mathcal{B}}

\newcommand{\cG}{\mathcal{G}}
\newcommand{\G}{\mathcal{G}}
\newcommand{\M}{\mathcal{M}}

\newcommand{\cL}{\mathcal{L}}
\newcommand{\cM}{\mathcal{M}}
\newcommand{\cQ}{\mathcal{Q}}
\newcommand{\cR}{\mathcal{R}}
\newcommand{\cU}{\mathcal{U}}

\newcommand{\cT}{\mathcal{T}}

\newcommand{\ba}{\mathbf{a}}

\newcommand{\bby}{\mathbf{y}}
\newcommand{\s}{\mathbf{s}}
\newcommand{\rr}{\mathbf{r}}
\newcommand{\argmin}{\operatornamewithlimits{argmin}}
\newcommand{\argmax}{\operatornamewithlimits{argmax}}
\newcommand{\PROB}{\mathop{\mathbb{P}}}
\newcommand{\bx}{\mathbf{x}}

\newcommand{\MM}{\mathcal{M}_M}

\newcommand{\Lap}{\mathrm{Lap}}
\newcommand{\acc}{\alpha}

\newcommand{\E}{\mathbb{E}}
\newcommand{\actions}{k}
\newcommand{\actionset}{\{1,2,\dots,\actions\}}
\newcommand{\dists}[1]{\Pi_{#1}}
\newcommand{\actiondists}{\dists{\actions}}
\newcommand{\Losses}{L}
\newcommand{\losses}{l}
\newcommand{\loss}{l}

\newcommand{\tlosses}{\losses_t}

\newcommand{\tjloss}{\loss_{t}^j}

\newcommand{\itjloss}{\loss_{i,t}^{j}}
\newcommand{\itjvloss}{\loss^{j}_{i,t,\hat\tau}} 	%%%%%%%%%% CHANGED v to \hat\tau
\newcommand{\state}{\pi}

\newcommand{\tstate}{\state_{t}}
\newcommand{\tjstate}{\state_{t}^j}

\newcommand{\itstate}{\state_{i,t}}

\newcommand{\rounds}{T}
\newcommand{\round}{t}
\newcommand{\scaled}[1]{\overline{#1}}
\newcommand{\scaledLosses}{\scaled{\Losses}}
\newcommand{\noisy}[1]{\widehat{#1}}
\newcommand{\noisyLosses}{\noisy{\Losses}}
\newcommand{\noisytjloss}{\noisy{\loss}_{t}^j}
\newcommand{\noisyitjvloss}{\noisy{\loss}^{j}_{i,t,\hat\tau}} %%%%% CHANGED v to \hat\tau
\newcommand{\Noises}{Z}
\newcommand{\noise}{z}
\newcommand{\tjnoise}{z_{t}^j}
\newcommand{\itjnoise}{z^j_{i,t}}

\newcommand{\noisyitjloss}{\noisy{\loss}^{j}_{i,t}}
\newcommand{\noisyitlosses}{\noisy{\losses}_{i,t}}
\newcommand{\noisylosses}{\noisy{\losses}}

\newcommand{\nralg}{\mathcal{A}}
\newcommand{\fixedalg}{\nralg_{\mathsf{fixed}}}
\newcommand{\swapalg}{\nralg_{\mathsf{swap}}}

\newcommand{\exploss}{\Lambda}
\newcommand{\regret}{\rho}
\newcommand{\mregret}{\regret_{\mathsf{max}}}
\renewcommand{\mod}{f}
\newcommand{\mods}{\mathcal{F}}
\newcommand{\somemods}[1]{\mods_{\mathsf{#1}}}
\newcommand{\fixedmods}{\somemods{fixed}}
\newcommand{\swapmods}{\somemods{swap}}

\newcommand{\modstate}{\mod \!\circ\! \state}

\newcommand{\set}[1]{\left\{#1\right\}}

\newcommand{\from}{\colon}

\newcommand{\eps}{\varepsilon}

\newcommand{\bits}{\{0,1\}}

\DeclareMathOperator*{\Expectation}{\mathbb{E}}
\newcommand{\Ex}[2]{\Expectation_{#1}\left[#2\right]}
\DeclareMathOperator*{\Probability}{\mathbb{P}}
\newcommand{\prob}[1]{\Probability\left[#1\right]}
\newcommand{\Prob}[2]{\Probability_{#1}\left[#2\right]}

\newcommand{\N}{\mathbb{N}}
\newcommand{\R}{\mathbb{R}}

\newcommand{\ignore}[1]{}
\newcommand{\anote}[1]{\ignore}
\newcommand{\jnote}[1]{\ignore}
\newcommand{\mpnote}[1]{\ignore}
\newcommand{\mknote}[1]{\ignore}
\newcommand{\rrnote}[1]{\ignore}

\newcommand{\err}{\mathbf{E}}

\newcommand{\pst}{\texttt{PST}}
\newcommand{\bm}{\texttt{BM}}
\newcommand{\abm}{\texttt{ABM}}
\newcommand{\pbr}{\texttt{P-BR}}
\newcommand{\nrl}{\texttt{NR-LAPLACE}}
\newcommand{\nrm}{\texttt{NRMEDIAN}}
\newcommand{\nrms}{\texttt{NRMEDIAN-SHARED}}

\begin{document}

\title{\textsc{Robust Mediators in Large Games}\thanks{We gratefully acknowledge the support of NSF Grant CCF-1101389 and the Alfred P. Sloan Foundation. We thank Nabil Al-Najjar, Eduardo Azevdeo, Eric Budish, Tymofiy Mylovanov, Andy Postlewaite, Al Roth, Tim Roughgarden, Ilya Segal, and Rakesh Vohra for helpful comments and discussions.
}}

\author{
Michael Kearns\thanks{Department of Computer and Information Science, University of Pennsylvania.} \and
Mallesh M. Pai\thanks{Department of Economics, University of Pennsylvania.} \and
Ryan Rogers\thanks{Department of Applied Mathematics and Computational Sciences, University of Pennsylvania.} \and
Aaron Roth\thanks{Department of Computer and Information Science, University of Pennsylvania.} \and
Jonathan Ullman\thanks{College of Computer and Information Science, Northeastern University.  Most of this work was done while the author was in the School of Engineering and Applied Sciences at Harvard University.}}

\maketitle
\begin{abstract}
A mediator is a mechanism that can only suggest actions to players, as a function of all agents' reported types, in a given game of incomplete information. We study what is achievable by two kinds of mediators, ``strong'' and ``weak.'' Players can choose to opt-out of using a strong mediator but cannot misrepresent their type if they opt-in. Such a mediator is ``strong'' because we can view it as having the ability to verify player types. Weak mediators lack this ability--- players are free to misrepresent their type to a weak mediator.  We show a striking result---in a prior-free setting, assuming only that the game is large and players have private types, strong mediators can implement approximate equilibria of the complete-information game. If the game is a congestion game, then the same result holds using only weak mediators. Our result follows from a novel application of \emph{differential privacy}, in particular, a variant we propose called  \emph{joint differential privacy}.
\end{abstract}

\newpage

\setcounter{tocdepth}{1}
\tableofcontents
\vfill
\newpage

%%%%%%%%%%%%%%%%%%%%%%%%%%%%%%%%%%%%%%%%%%%%%%%%%%%%%%%%%%%%%%%%%%%%%%%%%%%%%%
%%%%%%%%%%%%%%%%%%%%%%%%%%%%%%%%%%%%%%%%%%%%%%%%%%%%%%%%%%%%%%%%%%%%%%%%%%%%%%
\section{Introduction}
%%%%%%%%%%%%%%%%%%%%%%%%%%%%%%%%%%%%%%%%%%%%%%%%%%%%%%%%%%%%%%%%%%%%%%%%%%%%%%
%%%%%%%%%%%%%%%%%%%%%%%%%%%%%%%%%%%%%%%%%%%%%%%%%%%%%%%%%%%%%%%%%%%%%%%%%%%%%%
Mechanism design generally supposes a principal who can greatly influence the ``rules of the game'' among agents, to achieve a desirable outcome. In several applied settings of interest, however, this sort of influence is too much to hope for. For example, a market designer may have to live within the constraints of an existing market. She may be unable to guarantee that players participate in her designed mechanism, or that they follow the outcomes proposed even if they do participate. Similarly, in many settings, transfers may be ruled out.%
\footnote{A leading example is the central matching clearinghouses proposed for matching markets such as school choice, residency markets etc. Agents are free to drop out and contract on their own before the match, or indeed may deviate from the proposed match, and there are no transfers. A more recent example, closer to our model concerns the design of traffic-aware smartphone routing apps. Again, agents are free to not use the app, or may even try to mislead the app to get favorable routes. For a recent case study, see \url{http://www.wsj.com/articles/in-l-a-one-way-to-beat-traffic-runs-into-backlash-1447469058}.}
 To distinguish these principals of limited power from the standard principals considered in the literature, we refer to them as ``mediators.'' In this paper, we study what sorts of outcomes are achievable by mediators.

To be more precise, given a game of incomplete information, a mediator may only collect reports from players, and then \emph{suggest} actions to each player.  That is, we augment a given game with an additional option for each player to use a mediator. If a player opts in to using the mediator, and reveals her type, then it will suggest an action for that player to take (as a possibly stochastic function of all reports it received). We further require that mediators be robust or prior-free, i.e. that the mediator's recommendations (and the incentives of players to follow them) do not rely on the knowledge (or even existence) of any prior distribution over types.

Our main results are to show that it is possible to construct a robust mediator under very mild assumptions, namely, that the game is large and that the setting is one of private values. For every profile of realized types, our mediator suggests to each player her action in an approximate equilibrium of the realized full-information game. Further, we show that opting in to using the mediator, and following its recommendation forms an approximate \emph{ex-post} Nash equilibrium of the mediated game. The error of the approximation vanishes as the number of players grows.

The techniques used in our construction may be of independent interest. We propose and use a novel variant of the influential notion of \emph{differential privacy}, which we term \emph{joint differential privacy}.  In particular, we devise an algorithm for  computing equilibria of a full-information game such that any single player's reported type to the algorithm only has a small effect on the distribution of suggested actions to all other players, in the worst case over the reports of all players (it is this worst case guarantee that yields prior-free robustness). Ideas from the differential privacy literature allow us to precisely control the amount by which any single agent's type influences the suggested actions to others, and hence allows us to ensure our desired incentive properties. In other words, we prove our incentive guarantees for our mechanism by showing that it is differentially private, a technique that may be useful in other settings.

In this paper, we consider two kinds of mediators, which we call \emph{weak}, and \emph{strong mediators}, that extend the original game.  For both types of mediators, players may simply opt out of using the mediator and play the original game using any strategy they choose. Similarly, for both types of mediators, if they opt in to using the mediator, they are under no obligation to follow the action suggested to them, and can use the suggestion to act however they like. If the mediator is weak, then agents can also arbitrarily misrepresent their type to the mediator. In contrast, if players opt in to using a strong mediator, then they cannot misrepresent their type (equivalently, strong mediators have the ability to verify a player's type given that she opted in, which is why we term them ``strong''). In all cases, we assume that the mediator has very limited power---it cannot modify payoffs of the game (i.e. make payments) or enforce that any player take any particular action.

Under only the assumptions that the original game is ``large,'' and that types are private, we show that it is possible to implement a correlated equilibrium of the complete-information game using a strong mediator.\footnote{To be precise, we only guarantee that our mediator will select \emph{some} approximate equilibrium of the realized continuum, \emph{not} that the designer has control over what equilibrium is selected. } Informally, a game is large if there are many players and any player individually has only a small effect on the utility of any other player. Remarkably, no further assumptions on the game (e.g. single crossing properties) are required. We show that in such games there exists a strong mediator such that it is an approximate ex-post Nash equilibrium for every player to opt in to using the mediator, and then faithfully follow the suggested action of the mediator. In particular, this behavior forms a Bayes-Nash equilibrium for any prior on agent types, since it is an ex-post Nash equilibrium. Moreover, when players follow this strategy, the resulting play forms an approximate correlated equilibrium of the complete-information game.

We also show that for more structured large games---in particular, congestion games---it is possible to implement an  approximate Nash equilibrium of the complete-information game using only a weak mediator. Congestion games are examples of potential games, introduced by \citeasnoun{MS96}. They describe them thus:
\begin{quote}
``The class of congestion games is, on the one hand, narrow, but on the
other hand, very important for economics. Any game where a collection
of homogeneous agents have to choose from a finite set of
alternatives, and where the payoff of a player depends on the number
of players choosing each alternative, is a congestion game.'' \cite{MS96}
\end{quote}

In fact, they show that the class of potential games is isomorphic to the class of congestion games.\footnote{However, the mapping from potential games to congestion games they exhibit yields a game with a number of actions that is exponential in the number of players in the game. Our results in contrast require that the number of alternatives be small relative to the number of players, and so we do not obtain positive results for the entire class of potential games, but only for those potential games that can be expressed \emph{concisely} as congestion games.} For congestion games, we construct a weak mediator such that it is an approximate ex-post
Nash equilibrium for every player to opt into using the mediator, truthfully
report their type, and then faithfully follow the mediator's suggested
action, and when players do so, the resulting play forms an approximate Nash
equilibrium of the complete-information game. In both of these results, the approximation error tends to $0$ as the number of players grows.

A tempting approach to obtaining our results is to consider the following mediator. The mediator
accepts a report of each agent's type, which defines an instance of a
complete-information game. It then computes an equilibrium (either Nash or
correlated) of the full-information game, and suggests an action to each
player which is a draw from this equilibrium distribution. By definition of
an equilibrium, if all players opt-in to the mediator and truthfully report
their type, then they can do no better than subsequently following the
recommended action. However, this mediator does not solve the problem, as it
may not be in a player's best interest to opt in, even if the other $n-1$
players do opt in! Intuitively, by opting out, the player can cause the
mediator to compute an equilibrium of the wrong game, or to compute a
different equilibrium of the same game. More specifically, by opting out of the mechanism, a player can have a substantial effect on the computed equilibrium, even if each player has only small effect on the utilities of other players.  That is, in the absence of other guarantees, the mediator itself might be quite unstable to individual agent reports. A similar argument applies to misreports to a weak mediator.

The task, then, is to find a stable way to compute equilibria of the underlying game that is insensitive to any single player's report. We show that this can be done for approximate correlated equilibria of any large game, and approximate Nash equilibria of large congested games. Specifically, in Section \ref{sec:truthfulness}, we argue that algorithms which (1) are jointly differentially private, and (2) compute approximate equilibria of the underlying complete information game provides appropriate incentives. Then in Sections \ref{sec:weakmediator} and \ref{sec:strongmediatorlaplace} we provide an algorithm for computing an approximate Nash equilibrium in large congestion games subject to joint differential privacy and an algorithm for computing an approximate correlated equilibrium subject to joint differential privacy in any large game, respectively.

In Section \ref{sec:discussion} we argue that our results for general large games are tight in several senses we make precise. In particular, we show that, in general, (1) there do not exist mediators that implement \emph{exact} rather than approximate equilibria, (2) it is not possible to achieve \emph{more exact} equilibria than what we achieve without employing substantially different techniques, and (3) that mediators that provide strong welfare guarantees are computationally infeasible.

We argue that in certain games, our mediators achieve good (utilitarian) social welfare.  Finally, we discuss why privately computing a correlated equilibrium is not sufficient to implement a weak mediator, even though it is sufficient to implement a strong mediator, and why instead it is necessary to privately compute a Nash equilibrium in order to implement a weak mediator.

%%%%%%%%%%%%%%%%%%%%%%%%%%%%%%%%%%%%%%%%%%%%%%%%%%%%%%%%%%%%%%%%%%%%%%%%%%%%%%%
\subsection{Related Work and Discussion} \label{ssec:relwork}
%%%%%%%%%%%%%%%%%%%%%%%%%%%%%%%%%%%%%%%%%%%%%%%%%%%%%%%%%%%%%%%%%%%%%%%%%%%%%%

\paragraph{Market and Mechanism Design}
Our work is related to the substantial body of literature on mechanism/market design in ``large games.'' In this literature, the largeness enables mechanisms to have good incentive properties, even when the small market versions do not. It stretches back to \citeasnoun{roberts1976incentives} who showed that market (Walrasian) equilibria are approximately strategy proof in large economies. More recently \citeasnoun{immorlica2005marriage}, \citeasnoun{kojima2009incentives}, \citeasnoun{kojima2010matching} have shown that various two-sided matching mechanisms are approximately strategy proof in large markets. There are similar results in the literature for one-sided matching markets, market economies, and double auctions. The most general result is that of \citeasnoun{azevedo2011strategyproofness} who design incentive compatible mechanisms for large economies that satisfy a smoothness assumption.

One salient aspect of our work in relation to most of this literature is that we make substantially less restrictive assumptions. We do not assume anything about how player types are generated (e.g. we make no distributional assumptions, and do not assume a replication economy), we do not assume anything about the structure of equilibria in the games we study, and our results give non-trivial bounds when the number of players is small (our bounds are non-trivial even when the number of players is much smaller than the size of the type space, for example).

Inspired by applied market design settings such as the medical resident match, other papers have also observed that a designer may have limited ability to affect/ limit participants' actions in a pre-existing market, and this should be taken into account in the design---see e.g. \citeasnoun{rothshorrer}.

\paragraph{Large Games} Our results hold under a ``largeness condition'', i.e. a player's unilateral deviation can affect the payoff of others by only a small amount. This is closely related to the literature on large games, see e.g. \citeasnoun{al2000pivotal} or \citeasnoun{kalai2004large}. There has been recent work studying large games using tools from theoretical computer science, studying robustness of equilibrium concepts, see \citeasnoun{gradwohl2008fault, gradwohl2010partial}. Our results can be interpreted as a structural result about large normal form games: there exists a family of approximate equilibria in all large games that are stable to changing any single player's payoffs in the game (the joint distribution of all other players' actions does not change much).

\paragraph{Communication}
There has been some previous work on mediators in games~\citeasnoun{mediators1,mediators2}. This line of work aims to modify the equilibrium structure of full-information games by introducing a mediator, which can coordinate agent actions if they choose to opt in using the mediator. Mediators can be used to convert Nash equilibria into dominant strategy equilibria, or implement equilibria that are robust to collusion. While the setting is substantially different, qualitatively, our mediators are weaker in that they cannot make payments \cite{mediators1}, or enforce actions \cite{mediators2}.

Our mediators are closer to the communication devices in the ``communication equilibria'' of  \citeasnoun{forges1986approach}. Her mediators are prior dependent, and she focuses on characterizing the set of payoffs achievable via such mediators. No attention is paid to computational tractability.

As was recently noted by \citeasnoun{CKRW15},\footnote{They extend our techniques to design weak mediators for aggregative games. These are games in which every player's payoff is a function of the player'€™s own action and a low-dimensional aggregate of all players' actions.} the results reported here and in that work
can be viewed as giving robust versions of the
communication equilibria of \citet{forges1986approach}, in the sense that we can implement an approximate equilibrium of the
complete information game using a ``direct revelation mediator,'' but
without needing the existence of a prior distribution on types. Compared
to the former, in which truthful reporting is a Bayes-Nash Equilibrium, truth-telling here
forms an ex-post Nash equilibrium. We include a summary
of this view of our results
in Table \ref{table:truthful}.

\begin{table}[ht]

\begin{small}
\begin{tabular}{|c|c|c|c|c|c|}
\hline
  {\bf Mechanism} & {\bf Class of} & {\bf Common} & {\bf Mediator} & {\bf Solution} &{ \bf Ex-post}\\
   		     &    {\bf Games} & {\bf Prior?} & {\bf Strength} & {\bf Concept} &{\bf Realization}\\
\hline
\citeasnoun{forges1986approach} & Any & Yes & Weak &  Bayes-Nash & ---\\[15pt]
This work & Large & No & Strong & Ex-post & Correlated\\
&&&&Nash&Eqbm \\[10pt]
This work & Large & No & Weak & Ex-post & Nash \\
& Congestion & && Nash& Eqbm\\[10pt]
\citet{CKRW15} & Aggregative  & No & Weak & Ex-post & Nash \\
&&&&Nash&Eqbm \\[5pt]
\hline
\end{tabular}
\end{small}
% \caption{Summary of truthful mechanisms for various classes of games
% and solution concepts. In the ``mediated game'', truth telling forms
% an ex-post Nash equilibrium in \cite{KPRU14,RR14}, as well as in our
% work. It forms a Bayes-Nash equilibrium using the revelation
% principle. A ``weak'' mediator does not require the ability to
% verify player types. A ``strong'' mediator does. Weak mediators are
% preferred. }
\caption{Summary of truthful mechanisms for various classes of games
  and solution concepts. Note that a ``weak'' mediator does not
  require the ability to verify player types. A ``strong'' mediator
  does.}
\label{table:truthful}
\end{table}

\paragraph{Differential Privacy}
Differential privacy was first defined by \citeasnoun{DMNS06}, and is now the standard privacy ``solution concept'' in the theoretical computer science literature. It was introduced to quantify the \emph{worst-case harm} that can befall an individual from allowing his data to be used in a computation, as compared to if he did not provide his data. It is useful in our context because it provides a stability guarantee on the distribution of outputs induced by a mechanism with respect to small changes in its inputs. There is by now a very large literature on differential privacy, readers can consult \citeasnoun{DR13} for a more thorough introduction to the field.  Here we mention work at the intersection of privacy and game theory, and defer a longer discussion of related work in the privacy literature to Appendix \ref{app:related}.

\citeasnoun{MT07} were the first to observe that a differentially private algorithm is also approximately dominant strategy truthful. This line of work was extended by \citeasnoun{NST12} to give mechanisms in several special cases which are exactly truthful by combining private mechanisms with non-private mechanisms which explicitly punish non-truthful reporting. \citeasnoun{HK12} showed that the mechanism of \citeasnoun{MT07} can be made exactly truthful with the addition of payments. This connection between differential privacy and truthfulness does not carry over to joint-differential privacy, but as we show in Section \ref{sec:truthfulness}, it is regained if the object that we compute privately is an equilibrium of the underlying game.

%%%%%%%%%%%%%%%%%%%%%%%%%%%%%%%%%%%%%%%%%%%%%%%%%%%%%%%%%%%%%%%%%%%%%%%%%%%%%%
%%%%%%%%%%%%%%%%%%%%%%%%%%%%%%%%%%%%%%%%%%%%%%%%%%%%%%%%%%%%%%%%%%%%%%%%%%%%%%
\section{Model and Preliminaries} \label{sec:prelims}
%%%%%%%%%%%%%%%%%%%%%%%%%%%%%%%%%%%%%%%%%%%%%%%%%%%%%%%%%%%%%%%%%%%%%%%%%%%%%%
%%%%%%%%%%%%%%%%%%%%%%%%%%%%%%%%%%%%%%%%%%%%%%%%%%%%%%%%%%%%%%%%%%%%%%%%%%%%%%
We consider games $\mathcal{G}$ of up to $n$ players $\{1,2,\ldots, n\}$, indexed by $i$. Player $i$ can take actions in a set $A_i$, $|A_i| \leq k$, where actions are indexed by $j$. To allow our games to be defined also for fewer than $n$ players, we will imagine that the null action $\bot \in A_i$, which corresponds to ``opting out'' of the game.  We define $A \equiv \prod_{i = 1}^n A_i$.  A tuple of actions, one for each player, will be denoted $\ba =(a_1, a_2, \ldots a_n) \in A$.%\footnote{In general, subscripts will refer to indices (e.g.~players and periods), while superscripts refer to components of vectors.}

Let $\mathcal{T}$ be the set of player types.\footnote{It is trivial to extend our results when agents have different typesets, $\mathcal{T}_i$.  $\mathcal{T}$ will then be $\bigcup_{i=1}^{n} \mathcal{T}_i$.} There is a utility function $u:\mathcal{T} \times A \rightarrow \R$ that determines the payoff for a player given his type $\tau_i$ and a joint action profile $\ba$ for all players. When it is clear from context, we will refer to the utility function of player $i$, writing $u_i:A\rightarrow \R$ to denote $u(\tau_i,\cdot)$.  We are interested in implementing an equilibrium of the realized \emph{complete-information game} in settings of \emph{incomplete information}. In the complete-information setting, the type $\tau_i$ of each player is fixed and commonly known to all players. In such settings, we can ignore the abstraction of types and consider each player $i$ simply to have a fixed utility function $u_i = u(\tau_i, \cdot)$. In models of \emph{incomplete} information, players know their own type, but do not know the types of others.

We restrict attention to ``large games,'' which are ``insensitive,'' in the following way.  Roughly speaking a game is $\lambda$-large if a player's choice of action affects any other player's payoff by at most $\lambda$. Note that we do not constrain the effect of a player's \emph{own} actions on his payoff---\emph{a player's action can have an arbitrary impact on his own payoff}. Formally:
\begin{definition}[$\lambda$-Large Game] \label{def:sensitive}
A game is said to be \emph{$\lambda$-large} if for any two distinct players $i\neq i'$, any two actions $a_i, a_{i}'$ and type $\tau_{i'}$ for player $i'$ and any tuple of actions $a_{-i},$
\begin{align}\label{eqn:sensitive}
|u(\tau_{i'}, (a_i,a_{-i})) - u(\tau_{i'}, (a_{i}',a_{-i}))| \leq \lambda.
\end{align}
\end{definition}

\subsection{Solution Concepts}
We now define the solution concepts we will use, both in the full information setting and in the incomplete-information setting.  Denote a distribution over $A$ by $\pi$, the marginal distribution over the actions of player $i$ by $\pi_i$, and the marginal distribution over the (joint tuple of) actions of every player but player $i$ by $\pi_{-i}$. We now present two standard solution concepts---(approximate) pure strategy Nash equilibrium and (approximate) correlated equilibrium.

\begin{definition}
(Approximate Pure Strategy Nash Equilibrium)  Let $(u_1, \cdots, u_n)$ be a profile of utility functions, one for each player.  The action profile $\mathbf{a}$ is an \emph{$\eta$-approximate pure strategy Nash equilibrium} of the (complete information) game defined by $(u_1, \cdots, u_n)$ if for any player $i$ and any deviation $a_i' \in A_i$, we have
$$
u_i(\mathbf{a}) \geq u_i(a_i',\ba_{-i}) - \eta.
$$
\end{definition}

\begin{definition}[Approximate Correlated Equilibrium] \label{def:ace}
Let $(u_1, \cdots, u_n)$ be a profile of utility functions, one for each player.  Let $\pi$ be a distribution over tuples of actions $A$.  We say that $\pi$ is an \emph{$\eta$-approximate correlated equilibrium} of the (complete information) game defined by $(u_1, \cdots, u_n)$ if for every player $i \in [n]$, and any function $f_i \from A_i \to A_i$,
\begin{equation*} \label{eqn:ace}
\Ex{\pi}{u_i(\ba)} \geq \Ex{\pi}{u_i(f(a_{i}), \ba_{-i})} - \eta
\end{equation*}
\end{definition}
We now define a desirable solution concept in the incomplete information setting.  Recall that in this setting, we make no assumptions about how player types are generated (i.e. they need not be drawn from a distribution).  We then define a player $i$'s strategy $s_i: \mathcal{T} \to A_i$ to be a mapping from her type to an action.

The following definition is a very strong solution concept, that we should not expect to exist in arbitrary games:
\begin{definition}[Ex-Post Nash Equilibrium] \label{def:expost}
The strategy profile $\s = (s_1, \cdots, s_n)$ is an \emph{$\eta$-approximate ex-post Nash equilibrium} if for every type profile $(\tau_1,\ldots,\tau_n) \in \mathcal{T}^n$, for every player $i$, and every unilateral deviation $a_i'\in A_i$, we have
$$
u(\tau_i, \s(\tau)) \geq  u(\tau_i, (a_i', \s_{-i}(\tau_{-i})) ) -  \eta
$$
\end{definition}
If it does exist, however, it is very robust. It is defined without knowledge of any prior distribution on types, but if there is such a distribution, then in particular, simultaneously for \emph{every} such distribution, an ex-post Nash equilibrium is also a Bayes Nash equilibrium.

For all of these notions of $\eta$-approximate equilibrium, if $\eta$ goes to $0$ as the number of players $n$ goes to infinity (i.e.~$\eta = o(1)$), then we say that we have an asymptotically exact equilibrium.

\subsection{Differential Privacy}
A key tool in our paper is the design of differentially private ``mediator'' algorithms for suggesting actions to play. Agents can opt out of participating in the mediator: so each agent can submit to the mediator some type $\tau_i$ (which may or may not be their actual type), or else a null symbol $\bot$ which represents opting out. A mediator is then a function from a profile of types (and $\bot$ symbols) to a probability distribution over some appropriate range $\cR^n$, i.e. $\cM \from (\mathcal{T} \cup \{\bot\})^n \to \Delta \cR^n$.

First we recall the definition of differential privacy, both to provide a basis for our modified definition, and since it will be a technical building block in our algorithms.  Roughly speaking, a mechanism is differentially private if for every input profile $\tau = (\tau_1, \cdots, \tau_n)$ and every $i$, %knowledge of the output $\cM(\tau)$ as well as $\tau_{-i}$ does not reveal ``much'' about $\tau_i$.
if player $i$ unilaterally changes his report, he has only a small effect on the \emph{distribution} of outcomes of the mechanism.
\begin{definition}[Differential Privacy~\cite{DMNS06}] \label{def:standardprivacy}
A mechanism $\cM$ satisfies \emph{$(\eps, \delta)$-differential privacy} if for any player $i$, any $\tau_i, \tau_{i}' \in \mathcal{T} \cup \{\bot\}$, any tuple of types for players $i' \neq i$, $\tau_{-i} \in (\mathcal{T} \cup \{\bot\})^{n-1}$ and any event $B \subseteq \cR^{n}$,
\begin{align*}
\Prob{\cM}{\left(\cM(\tau_i,\tau_{-i}) \right) \in B} \leq e^{\eps} \Prob{\cM}{\left(\cM(\tau_{i}', \tau_{-i})\right) \in B} + \delta.
\end{align*}
\end{definition}

Differential privacy is not quite suitable for our setting because for any ``reasonable'' notion of equilibrium, the action played by player $i$ must depend in a highly sensitive way on her own type (since she must be playing some sort of best response).  To circumvent this barrier, we define a relaxation that we call \emph{joint differential privacy.}
Roughly speaking, a mechanism is jointly differentially private if for every player $i$, the joint distribution on outcomes given to the $n-1$ players $j \neq i$ is stable with respect to player $i$'s report. It differs from differential privacy in that the outcome given to player $i$ can depend in an arbitrary way on player $i$'s own report.
\begin{definition}[Joint Differential Privacy] \label{def:privacy}
A mechanism $\cM$ satisfies \emph{$(\eps,\delta)$-joint differential privacy} if for any player $i$, any $\tau_i, \tau_{i}' \in \mathcal{T} \cup \{\bot\}$, any tuple of utilities for players $i' \neq i$, $\tau_{-i},$ and any event $B \subseteq \cR^{n-1}$,
\begin{align*}
\Prob{\cM}{\left(\cM(\tau_i, \tau_{-i}) \right)_{-i} \in B} \leq e^{\eps} \Prob{\cM}{\left(\cM(\tau_{i}',\tau_{-i})\right)_{-i} \in B} + \delta.
\end{align*}
\end{definition}

\subsection{Congestion Games}
As we suggested earlier, our results  corresponding to weak mediators hold only for the class of large congestion games.
An instance of a congestion game with $n$ players is given by a set of $m$ facilities $E$.  The set of actions available to a player in the game is a subset of $2^{E}$ which depends on her type. Further, there is a loss function $\ell_e: \R \to [0,1]$ for each facility that maps the number of people using that facility to a value between zero and one.

Players wish to minimize their cost, which is the sum of the losses on the facilities that they play on.  For a set of $n$ players with type profile $\tau = (\tau_i)_{i=1}^n$, we denote the vector of their chosen actions by $\mathbf{a} = (a_i)_{i = 1}^n \in A$ (recall that each action corresponds to a subset of the facilities).  For a given profile, we will denote the number of people choosing a facility $e\in E$ by $y_e(\ba)$, which for simplicity we will write as $y_e$ when the chosen actions are understood.  We now  define the cost of a player of type $\tau_i$ fixing a given action profile.
$$
c(\tau_i, \mathbf{a}) = \sum_{e\in a_i}\ell_e(y_e).
$$
Note that in the previous section we referred to utilities of each player, whereas here we refer to costs. This is merely notational--- we could equivalently define the utility for a player of type $\tau_i$ for a fixed action profile $\ba$ to be
$$
u(\tau_i,\ba) = \sum_{e \in a_i} (1-\ell_e(y_e)),
$$
which has an upper bound of $m$.

We define the sensitivity of a loss function to quantify how much a player can impact the loss on each facility by a unilateral deviation.
\begin{definition}(Sensitivity) The \emph{sensitivity} $\sigma_\ell$ for an $n$-player congestion game with facilities $E$ and loss functions $\ell = (\ell_e)_{e \in E}$ is defined as
$$
\sigma_\ell = \max_{e \in E; 0\leq y_e < n}\left| \ell_e(y_e+1) - \ell_e(y_e)\right|.
$$
\end{definition}
  Hence, a congestion game $\G$ with sensitivity $\sigma_\ell$ is $\lambda$-large where $\lambda = m\sigma_\ell$.

One common example of a congestion game is a ``traffic routing game.'', In such a game, the facilities $E$ are the edges of a directed graph $G = (V,E)$, with $|E| = m$.  The actions in the routing game correspond to paths in the graph, which are subsets of the edges.  The type $\tau$ of a player specifies a particular source and destination pair in the graph between which the player wishes to route traffic, and feasible actions for a player are those subsets of edges that form paths between the source and destination pair given by the player's type.  Recall that $A$ denotes the union of all action spaces: in traffic routing games, the set of all simple paths in $G$.
\subsection{Mediated Games}

The main result of this paper is a reduction that takes an arbitrary large game $\mathcal{G}$ of incomplete information and modifies it to have an ex-post Nash equilibrium implementing equilibrium outcomes of the corresponding full-information game defined by the \emph{realized} agent types.
Specifically, we modify the game by introducing the option for players to use a \emph{mediator} $\cM$ that can recommend actions to the players.  A mediator is defined to be an algorithm $\cM: (\mathcal{T}\cup \{\perp\})^n \to A$, which takes as input reported types (or $\bot$ for any player who declines to use the mediator), and outputs a suggested action to each player who opted in. We consider two cases, depending on whether the mediator is weak or strong.

For a weak mediator $\cM$, we construct a new game $\cG^W_\cM$ called the \emph{weakly mediated game} induced by the mediator $\cM$. Informally, in $\cG^W_\cM$, players have several options.  They can \emph{opt-out} of using the mediator (i.e. report $\perp$) and select an action independently of it. Alternately they can \emph{opt-in} and report to it some type (not necessarily their true type), and receive a suggested action $a_i$. We will denote the action set to a player of type $\tau_i$ as $A(\tau_i)$ or simply as $A_i$ for the $i$th player of type $\tau_i$.  Players are free to follow this suggestion or use it in some other way.  That is, they can play an action $f_i(a_i)$ for some arbitrary function $f_i:A(\tau_i') \rightarrow A_i$ where $\tau_i'$ is the reported type of player $i$. Formally, the game $\cG^W_\cM$ has an action set $A^W_i$ for each player $i$ defined as:
$$
A^1_{i} = \{(\tau_i', f_i): \tau_i' \in \mathcal{T}, f_i : A(\tau_i') \to A_{i}\} \qquad
A^2_i = \{ (\perp, f_i): f_i \text{ is constant} \},\\
$$
$$
A^W_i = A^1_i \cup A^2_i.
$$
We define $ A^W = \prod_{i=1}^n A^W_i$.  Mediators may be randomized, so we define the utility functions $ u^W_\cM$ in $\cG^W_\cM$ to be the expected utility that agents receive in the original game (defined by utilities $u: \mathcal{T} \times A \to [0,m]$ in $\G$) when players use functions $f_i$ to interpret the suggested actions provided by the mediator.  We will use the notation $\mathbf{f}(\mathbf{a}) = (f_1(a_1),\cdots, f_n(a_n) ). $
$$
 u^W_\cM:\mathcal{T}\times  A^W \to \R;  \quad \quad  u^W_\cM(\tau_i, (\tau',\mathbf{f})) =  \E_{\ba \sim \cM(\tau')}[u\left(\tau_i,\mathbf{f}(\ba)\right)].
$$

We also consider a stronger form of mediation, in which the mediator has the power to verify reported agent types (equivalently, agents do not have the ability to misreport their type to the mediator if they choose to opt in). To model such mediators, we limit the action space of $\cG^W_\cM$ so that players can either \emph{opt-out} of using the mediator $\cM$ (i.e. report type $\bot$) or \emph{opt-in} to using $\cM$, (which we denote as a report of type $\top$), which is equivalent to her reporting her true type.  Formally, given a game $\mathcal{G}$ defined by an action set $A$, a type space $\mathcal{T}$, and a utility function $u$, we define the \emph{strongly mediated game} $\cG^S_\cM$, parameterized by an algorithm $\cM:\{\mathcal{T}\cup \{\bot\}\}^n \rightarrow A$, as the game where player $i$ chooses an action from $A_i^S$ defined as
$$
 (A_i')^1 = \{(\top, f) | f_i:A_i\rightarrow A_i\} \qquad A^2_i = \{ (\perp, f_i): f_i \text{ is constant} \}
$$
$$
A^S_i= (A_i')^1 \cup A_i^2.
$$
We then define $A^S = \prod_{i=1}^n A^S_i$.  Given a set of choices by the players in a strongly mediated game $\cG^S_\cM$, we define a vector $\mathbf{x} = (x_i)_{i=1}^n$ such that $x_i = \tau_i$ for each player $i$ who chose $(\top, f_i)$ (each player who opted in), and $x_i = \bot$ for each player $i$ who chose $(\bot, a_i)$ (each player who opted out). The mediator then computes $\M(\mathbf{x}) = \ba $.  This results in a vector of actions $\ba$ from the original game $\mathcal{G}$, one for each player. For each player who opted in, they play the action $f_i(a_i)$. For each person who opted out, they play an action which was chosen independently of the mechanisms output.  The utility $u^S_\cM$ for the strongly mediated game $\cG^S_\cM$ is then
$$
u^S_\cM:\mathcal{T}\times A^S \to \R  \quad \quad u^S_\cM(\tau_i, \bx) =  \E_{\ba \sim \cM(\bx)}[u\left(\tau_i,\mathbf{f}(\ba)\right)].
$$

%%%%%%%%%%%%%%%%%%%%%%%%%%%%%%%%%%%%%%%%%%%%%%%%%%%%%%%%%%%%%%%%%%%%%%%%%%%%%%
%%%%%%%%%%%%%%%%%%%%%%%%%%%%%%%%%%%%%%%%%%%%%%%%%%%%%%%%%%%%%%%%%%%%%%%%%%%%%%
\section{Joint Differential Privacy and Incentives} \label{sec:truthfulness}
%%%%%%%%%%%%%%%%%%%%%%%%%%%%%%%%%%%%%%%%%%%%%%%%%%%%%%%%%%%%%%%%%%%%%%%%%%%%%%
%%%%%%%%%%%%%%%%%%%%%%%%%%%%%%%%%%%%%%%%%%%%%%%%%%%%%%%%%%%%%%%%%%%%%%%%%%%%%%
We have now defined two augmented games parametrized by an algorithm $\cM$ and in both, players can deviate from the suggested action of the mediator  $\cM$: the weakly mediated game $\cG^W_\cM$ allows players to falsely report types to the mediator, and the strongly mediated game $\cG^S_\cM$ does not (although it allows players to opt out of using the mediator and withhold their type).  We now show that if the algorithm $\cM$ satisfies certain properties then in the resulting mediated game it is an asymptotic ex-post Nash equilibrium for every player to truthfully report their types to the algorithm $\cM$ and follow the suggested action of $\cM$ in $\cG^W_\cM$ or opt-in to using $\cM$ and follow its suggestion in $\cG^S_\cM $.

\begin{theorem}[Weakly Mediated Games]
Let $\M$ be a mechanism satisfying $(\epsilon,\delta)$-joint differential privacy such that, when $\M$ is run on any input type profile $\tau$, with probability at least $1-\beta$,  $\M$ computes an $\eta$-approximate Nash equilibrium of the complete-information game defined by $\tau.$  Then it is an $\eta'-$ approximate ex-post Nash equilibrium of the weakly mediated game $\cG^W_\M$ for each player $i$ to play $(\tau_i, f_i)$ where $\tau_i$ is $i$'s actual type, $f_i$ is the identity function for each $i \in [n]$, and
$$\eta' = \eta + U(\epsilon + \beta +\delta) $$
for an upper bound $U$ on the utilities of every player. Informally, we will call the action profile $(\tau, \mathbf{f})$ ``good behavior.''
\label{lem:main}
\end{theorem}
\begin{proof}
We consider the strategy profile $(s_1, \cdots, s_n)$ where $s_i(\tau_i) = (\tau_i, f_i)$ is the action in $\cG^W_\M$ where player $i$ reports $\tau_i$ and then faithfully follows the suggestion of $\M$, i.e. $f_i$ is the identity function.  Consider a unilateral deviation for player $i$, $s_i'(\tau_i)= (\tau_i',f'_i)$ where $f_i': A(\tau_{i}') \to A_i$ is some function mapping actions feasible for reported type $\tau_i'$ to actions feasible for her actual type $\tau_i$.  We seek to bound the gain that she can obtain through this unilateral deviation:
$$
 u^W_\M(\tau_i,\s(\tau))  -  u^W_\M(\tau_i, (s_i'(\tau_i),\s_{-i}(\tau_{-i}))).
$$
Let $\mathbf{\tau}'$ be the type profile that results when player $i$ makes this unilateral deviation i.e. $\tau_{-i} = \tau_{-i}'$.  We define the best response action for a player of type $\tau_i$ given the action profile of the other players to be
$$
BR_{\tau_i}(\ba_{-i}) = \argmax_{ a_i \in A_i} \{u(\tau_i, (a_i,\ba_{-i}))\}.
$$
breaking ties arbitrarily.
We first condition on the event that $\M$ outputs an $\eta$-approximate Nash Equilibrium (which occurs with probability at least $1-\beta$). We then know,
\begin{align}
 u^W_\M(\tau_i,\s(\tau))&  = \Ex{\ba \sim \M(\tau)}{u(\tau_i,(a_i,\ba_{-i}) )} \nonumber\\
& \geq \Ex{\ba \sim \M(\tau)}{u(\tau_i, (BR_{\tau_i}(\ba_{-i}),\ba_{-i})) }-\eta .
\label{eq:ne_not_corr}
\end{align}
Next, we invoke the privacy condition:
\begin{align*}
 u^W_\M(\tau_i,\s(\tau))&  \geq \exp(-\epsilon) \Ex{(a_i',\ba_{-i}) \sim \M(\tau')}{u(\tau_i, (BR_{\tau_i}(\ba_{-i}),\ba_{-i})) }- \eta  -U\delta \\
& \geq  \Ex{(a_i',\ba_{-i}) \sim \M(\tau')}{  u(\tau_i, (BR_{\tau_i}(\ba_{-i}),\ba_{-i})) }  - \eta- U\epsilon- U\delta  \\
& \geq  \Ex{(a_i',\ba_{-i})\sim \M(\tau')} { u(\tau_i, (f'_i(a'_i),\ba_{-i})) }- \eta - U\epsilon-U\delta \\
& =  u^W_\M(\tau_i, (s_i', \s_{-i}(\tau_{-i}))) -\eta- U\epsilon - U\delta.
\end{align*}
The first inequality follows from joint differential privacy, the second from the inequality $\exp(-\epsilon) \geq 1- \epsilon$ together with the fact that $U$ is an upper bound on the utility function, and the third from the definition of the best response.  Combining this analysis with the $\beta$-probability event that the mechanism fails to output an approximate Nash equilibrium (and hence any player might be able to improve his utility by $U$) gives the result:
\begin{equation*}
 u^W_\M(\tau_i,\s(\tau)) \geq  u^W_\M(\tau_i, (s_i', \s_{-i}(\tau_{-i}))) - \eta -U\epsilon - U \delta-U\beta. \qedhere
\end{equation*}\end{proof}
%The following corollary comes from the fact that the previous theorem applied for any realized type vector $\s$ and thus will apply to any distribution $\tau$ over types.
%\begin{corollary}
%Let $\M$ be a mechanism satisfying $(\epsilon,\delta)$-joint differential privacy that on any input type profile $\s$ with probability $1-\beta$ computes an $\eta$-approximate pure strategy Nash Equilibrium of the complete information game defined by $\s$.  Then for every prior distribution $\tau$ over the type space $\mathcal{T}$, it is an $\eta'-$ approximate Bayes Nash Equilibrium of $\ctG_\M$ for each player to have ``good" behavior, and $\eta' =  \eta+ L\epsilon +L \beta+ L\delta$.
%\end{corollary}

Because we cannot compute Nash equilibria efficiently in general games, and because there are settings in which we know how to privately compute correlated equilibria, but do not know how to privately compute Nash equilibria, we would like to know what is implied when a mediator computes a correlated equilibrium.

\begin{theorem}[Strongly Mediated Games]
Let $\M$ be a mechanism satisfying $(\epsilon,\delta)$-joint differential privacy such that, on any input type profile $\tau,$ with probability at least $1-\beta$, $\M$ computes an $\eta$-approximate correlated equilibrium of the complete-information game defined by $\tau$.  Then it is an $\eta'-$ approximate ex-post Nash equilibrium of the strongly mediated game $\cG^S_\M$ for each player $i$ to play $(\top , f_i)$ where $f_i$ is the identity function for each $i \in [n]$, and
$$\eta' = \eta + U(\epsilon + \delta +\beta) $$
where $U$ is an upper bound to the utilities of every player. We will call the action profile $(\top, \mathbf{f})$ ``nice" behavior.
\label{thm:truth}
\end{theorem}
\begin{proof}
We consider the strategy profile $(s_1, \cdots, s_n)$ where $s_i(\tau_i) = (\top, f_i)$ the action in $\cG_\M^S$ where player $i$ opts-in to using the mediator $\M$ and $f_i$ is the identify function.  There are two types of deviations that a player $i$ can consider: $s_i'(\tau_i) = (\top, f'_i)$ for some function $f'_i:A_i\rightarrow A_i$ not the identity function, and $s_i''(\tau_i)=(\bot, a_i')$ for some action $a'_i$. First, we consider deviations of the first kind, $s_i'$.  Conditioning on the event that $\M$ outputs an $\eta$ approximate correlated equilibrium of $\G$ (which occurs with probability at least $1-\beta$):
\begin{align*}
u^S_\M(\tau_i,\s(\tau)) &= \Ex{\ba \sim \M(\tau)}{u(\tau_i,\ba)} \\
&\geq \Ex{\ba \sim \M(\tau)}{u(\tau_i,(f'_{i}(a_{i}), \ba_{-i})} - \eta \\
&= u^S_\M(\tau_i,(s'_i(\tau_i), \s_{-i}(\tau_{-i}))) - \eta
\end{align*}
where the inequality follows from the fact that $\M$ computes an $\eta$-approximate correlated equilibrium. Now, consider a deviation of the second kind $s_i''(\tau_i) = (\bot,a_i')$. Again, condition on the event that $\M$ outputs an $\eta$-approximate correlated equilibrium of $\G$:
\begin{eqnarray*}
u^S_\M(\tau_i,\s(\tau))  &=&  \Ex{\ba \sim \M(\mathbf{\tau})}{u(\tau_i,\ba)} \\
&\geq& \Ex{\ba \sim \M(\mathbf{\tau})}{u(\tau_i,(a_i',\ba_{-i}))} - \eta \\
&\geq& \exp(-\epsilon)\cdot\Ex{\ba \sim \M(\bot,\mathbf{\tau}_{-i})}{u(\tau_i,(a_i',\ba_{-i}))} - U\delta - \eta\\
&\geq& \Ex{\ba \sim \M(\bot,\mathbf{\tau}_{-i})}{u(\tau_i, (a_{i}',\ba_{-i}))} - U\epsilon - U\delta - \eta \\
&=& u^S_\M(\tau_i,(s_i''(\tau_i),\s(\tau))) -U\eps - U\delta - \eta
\end{eqnarray*}
where the first inequality follows from the $\eta$-approximate correlated equilibrium condition, the second follows from the $(\epsilon,\delta)$-joint differential privacy condition, and the third follows from the fact that for $\epsilon \geq 0$, $\exp(-\epsilon) \geq 1-\epsilon$ and that utilities for $\G$ are bounded by $U$.  Lastly, we consider the event where $\M$ fails to output an approximate correlated equilibrium, which may decrease the expected utility of any player by at most $U \beta$ (since utilities are bounded by $U$, and the event in question occurs with probability at most $\beta$.).  This completes the theorem.
\end{proof}

%We then get the following corollary
%\begin{corollary}
%Let $M$ be an algorithm that satisfies $(\epsilon,\delta)$-joint differential privacy, and be such that for every vector of types $\mathbf{s} \in \mathcal{T}^n$, $M(\mathbf{s})$ induces a distribution over actions that is an $\eta$-approximate correlated equilibrium of the full-information game $\mathcal{G}$ induced by the type vector $\mathbf{t}$.  Then for every prior distribution on types $\tau$, it is an $\eta'$-approximate Bayes Nash equilibrium of $\mathcal{G}'_M$ for every player to play $(\top, f)$ for the identity function $f(a) = a$. (i.e. for every player to opt into the proxy, and then follow its suggested action), where $\eta' = \eta + L\epsilon + L\delta$.
%\end{corollary}

 The main technical contributions of this paper is that we construct a weak mediator that satisfies the hypothesis of Theorem \ref{lem:main} in large congestion games and we develop a strong mediator that satisfies the hypothesis of Theorem \ref{thm:truth} for all large games.  We first state our main result in the case of weak mediators assuming  $\lambda = O(1/n)$. In the general case, $\lambda$ appears in our theorem as a parameter. Note that we will use the notation $\tilde O(g(\theta)) = O(g(\theta) \log^r(g(\theta)) )$ where $\theta$ is a vector of parameters, $g$ is a function of all the parameters, and $r$ is some constant.

\noindent\textsc{Theorem} (Informal): Let $\G$ be any $O(1/n)$-large congestion game with $m$ facilities 
and nonnegative facility loss functions bounded by $1$.  Then there is a mediator $\M$ such that ``good behavior'' forms an $\eta'$-approximate ex-post Nash equilibrium of the weakly mediated game $\G^W_\M$, for:
$$
\eta' =  \tilde O\left(  \frac{m^{3/4} \sqrt{\log(n)} }{n^{1/4}} \right).%O\left(\left(  \frac{m^5\log^{5/2}(n)}{n}\right)^{1/4}\right).
$$
Moreover, when players play according to this equilibrium, the resulting play forms an $\eta'$-approximate Nash equilibrium of the underlying complete-information game induced by the realized player types.
\begin{remark}
Note that $\eta' = o(1)$ as $n$ grows.  Hence, asymptotically, good behavior forms an exact ex-post Nash equilibrium of the weakly mediated game.
\end{remark}

We next state our main result for arbitrary large games (again, under the assumption that $\lambda = O(1/n)$ -- in general, $\lambda$ will appear in the theorem as a parameter).

\noindent\textsc{Theorem} (Informal): Let $\mathcal{G}$ be any $O(1/n)$-large game with nonnegative utilities bounded by $1$ and action set for each player with size bounded by $|A_i| \leq k$. Then there is a mediator $\M$ such that ``nice behavior" forms an $\eta'$-approximate ex-post Nash equilibrium of the strongly mediated game $\G^S_{\M}$, for
  $$\eta' = \tilde O\left( \frac{ \sqrt{k \log(n)} }{n^{1/4}} \right).$$
 %If we can take the mediator to be computationally inefficient, then there is another mediator $\M_2$ such that ``nice behavior" forms and $\eta'_2$- approximate ex-post Nash equilibrium of $\G^S_{\M_2}$:
 %$$\eta_2' = \tilde{O}\left(\left(\frac{\log^{2} k \cdot \log^{3}|\mathcal{T}|}{n}\right)^{1/4}\right).$$
Moreover, when players play according to this equilibrium, the resulting play forms an $\eta'$-approximate correlated equilibrium of the underlying complete-information game induced by the realized player types.

 \begin{remark}
Note that $\eta' = o(1)$ as $n$ grows.  Hence, asymptotically, nice behavior forms an exact ex-post Nash equilibrium of the strongly mediated game.
 \end{remark}

The mediators that form the basis of both of these results are computationally efficient. In the appendix, we give a computationally \emph{inefficient} mediator that gives an improved dependence on $k$ in the general case.

\section{Weak Mediators for Congestion Games} \label{sec:weakmediator}
In this section, we design an algorithm that computes approximate Nash equilibria in congestion games while satisfying joint differential privacy.

Our approach is to simulate ``best response dynamics,'' i.e. dynamics in which each player best responds to the actions of other players in the previous round. It is known that in congestion games, these dynamics converge to an approximate Nash equilibrium in time linear in $n$ for any fixed approximation factor. In what follows, we describe this simulation as if it was an actual process being carried out by the players. Of course, what we intend is a centralized algorithm (the mediator) which simulates this process as part of its computation.

The main idea is that since each player's cost in a congestion game is a function only of \emph{how many people} are playing on each facility, it is possible to implement best response dynamics by giving each player a view only of a counter, one for each facility, which indicates how many players are currently playing on that facility. To implement the algorithm privately, we replace these exact counters with private counters, which instead maintain a ``noisy'' inexact version of the player counts. In order to bound the scale of the noise that is necessary for privacy, we need to bound how many times any single player might make a best response over the course of the game dynamics.

This is exactly where we use the largeness condition on the game: we show that since single moves of \emph{other} players cannot significantly increase the cost of any player $i$ in a large game, once $i$ has made a best response, a substantial number of other players must move before $i$ can substantially improve again by deviating. In combination with an upper bound on the total number of moves that ever need to be made before convergence, this gives an upper bound on the total number of moves made by any single player -- that depends on the largeness parameter, but not on the number of players. Finally, we show that best response dynamics is tolerant to these low levels of error on the counters.
\\
\indent We go right into designing an algorithm that simulates best response dynamics via an online counter that satisfies differential privacy.  There will be a counter for each facility indicating how many players are using that facility.  The counter keeps track of a stream of bits indicating when players arrive and leave the facility (i.e. when players make unilateral deviations to actions that now include, or no longer include the facility). The stream of bits on a facility $e$, denoted by $\omega_e^{\leq t} = (\omega_e^j)_{j = 1}^t \in \{ -1,0,1\}^t$ represents the state of the facility at times $1$ through $t$. At each time step $t$, $\omega_e^t = 1$ if a player has unilaterally changed his action to one that now includes facility $e$, $\omega_e^t = -1$ if a player has changed his action to one that no longer includes facility $e$, or $\omega_e^t = 0$ if the player changing his action at time $t$ was either playing on facility $e$ both before and after $t$, or neither before nor after $t$. From this stream, we can compute the number of players using a facility $e$ at time $t$, which we denote as $y_e^t = \sum_{j=1}^t \omega_e^j$. We cannot represent this sum exactly while satisfying differential privacy: we instead use a modified version of the binary mechanism of \citeasnoun{Binary} (given in $\abm$ in Algorithm \ref{ABM}) to maintain an approximate estimate of the sum of bits in this stream, which in turn represents at each time period the approximate number of players using the facility.  Here, we will denote a (mean 0) and scale $b$ Laplacian random variable by $\Lap(b)$.  The original binary mechanism $\bm$ is given in the appendix.
\begin{algorithm}[h]
\caption{Partial Sum Table}\label{PST}
\begin{algorithmic}[0]
\INPUT : Stream $\omega^{\leq t} \in \{-1,0,1\}$, past partial count table $\hat \rr $, and privacy parameter $\epsilon'$.
\OUTPUT : Updated partial counts $\hat \rr$.
\Procedure {$\pst$} {$\omega^{\leq t}, 	\hat \rr ,\epsilon'$}
%\State Initialize $q^t \gets  0^{\lfloor \log(T)\rfloor \times T}$, i.e. a zero matrix where $t \leq T$.
\State Write $t = \sum_{j = 0}^{\lfloor \log(t)\rfloor}b_j(t) 2^j$ where $b_j(t) \in \{ 0,1\}$.
\State $ i \gets \min\{j:b_j(t) \neq 0 \} $
\For{$j = 0, \cdots, i$}
	\State $\rr(j, t/2^{j} ) \gets \sum_{\ell =  t/2^{j} }^t w^\ell $ %\label{eq:exact_table}
	\State $\hat{\rr}(j, t/2^{j}) \gets \rr(j, t/2^{j} ) + \text{Lap}(1/\epsilon')$
\EndFor
	
%\State $q^{t}(j,t)$
%\For{$t = 1, \cdots, T$}
%\State Write $t = \sum_{j = 0}^{\lfloor \log(t)\rfloor}b_j(t) 2^j$
%\State $ i \gets \min\{j:b_j(t) \neq 0 \} $
%\State $q_{ i}(t) \gets \sum_{j< i} q_j(t) + \omega_t$
%\For{$j= 0, \cdots,i -1$}
%	\State $q_j(t+1) \gets 0, \quad \quad \hat q_j(t+1) \gets 0$
%\EndFor
%\State $\hat q_{ i} \gets q_{ i} +\text{Lap}(1/\epsilon') $
%%\State $\hat v_t \gets \hat q_i$
%\State $\hat{y}^t \gets \sum_{j:\{ b_j(t)=1\}} \hat{q}_{j}$
%\EndFor

\Return $\hat \rr$.
\EndProcedure
\end{algorithmic}
\end{algorithm}
%%%%%

We initialize a partial sum table $\hat \rr^1_e$ to be a table of zeros of size $\lfloor\log(nT) \rfloor + 1$ by $nT$ for each facility $e \in E$.  We update as many as $\log(nT) + 1$ entries of past partial sum table $\hat\rr_e^{t-1}$ to get $\hat\rr^t_e=\pst(\omega_e^{\leq t}, \hat\rr_e^{t-1},\epsilon')$.  We then use the partial sum table $\hat\rr_e^t$ to find the count $\hat y^t_e = \abm(\hat\rr_e^t,t)$ for round $t$.

%%%%%
\begin{algorithm}[h]
\caption{Adaptive Binary Mechanism}\label{ABM}
\begin{algorithmic}[0]
\INPUT : Partial Sum Table $\hat \rr$ and time $t \leq $ (length of the first row of $\hat \rr$).
\OUTPUT : A count $\hat y^t$.
\Procedure {$\abm$} {$\hat \rr,t$}
\State Write $t = \sum_{j = 0}^{\lfloor \log(t)\rfloor}b_j(t) 2^j$
\State $\hat y^t \gets \sum_{j=0}^{\lfloor \log(t) \rfloor} b_j(t) \hat \rr(j,t/2^{j}) $\\
\Return $\hat y^t$.
\EndProcedure
\end{algorithmic}
\end{algorithm}
%%%%%%%

At each ``time step'', a single player has the option of making a best response move (but will not if she has no move which substantially improves her utility). At time $t$ with action profile $\ba^t$, we compute the noisy count $\hat{y}_e^t = \hat y_e(\ba^t)$ on facility $e$ from the output of $\abm$. At the same time, we denote the exact counts on the facilities as $y_e^t=y_e(\ba^t)$.  We define the noisy costs for each player in terms of the counts $\hat \bby^t = (\hat{y}_e(\ba^t))_{e\in E}$ as
$$
\hat{c}(\tau_i, \ba^t) = \sum_{e \in a_i}\ell(\hat{y_e}(\ba^t) ).
$$
Players will make best responses on rounds when they can substantially decrease their cost according to the \emph{noisy} cost estimates.
\begin{definition}
($\alpha$-Noisy Best Response) Given an action profile $\ba^t$ of a congestion game $\G$, and noisy facility counts $\hat\bby^t$, an $\alpha$-noisy best response $a_i^*$ for player $i$ of type $\tau_i \in \mathcal{T}$ is
$$
a_i^* \in \argmin_{a_i' \in A_i} \left\{ \hat{c}(\tau_i,(a_i',a^t_{-i})): \hat c(\tau_i,\ba^t )- \hat c(\tau_i,(a_i',a_{-i}^t))\geq \alpha\right\}
$$
\end{definition}
We now use the result of \citeasnoun{Binary} which gives a bound on how much the noisy count of players on a facility $y_e^t$ and $\hat y_e^t$ can differ on a stream $\omega_e^{\leq t}$ of length $t \in [T]$ where $\hat y_e^t$ is computed via the binary mechanism $\bm$.
\begin{lemma}(Utility \cite{Binary})
\label{lem:chanerror}
For a facility stream $\omega_e^{\leq T}$ of length $T$, and for $t \in [T]$, let $y_e^t = \sum_{j=1}^t \omega_e^j$, and let $\hat y_e^t$ denote the estimated partial sum computed via the binary mechanism $\hat y_e^{\leq T } = \bm(\omega_e^{\leq T}, \epsilon')$. Then with probability at least $1-\beta'$ for $\beta' \in (2/T,1]$:
$$
\max_{t \in [T]}|\hat{y}_e^t - y_e^t|  \leq \frac{\sqrt{8\log(T)\log(2/\beta')}}{\epsilon'}.
$$
\end{lemma}
%In our setting, each bit $\omega_e^t$ depends on previous calculations done with the stream $(\omega_e^{\leq t-1})_{e \in E}$.  This is why we introduce a modified version of the binary mechanism to deal with an adaptive stream.  If we are given the stream, then the adaptive procedure $\abm$ produces the same distribution over counts as does the original $\bm$ procedure in \cite{Binary}.  That is if $\hat q^0$ is a $\log(T)$ by $T$ table with zeros in every entry, the stream $\omega_e^{\leq T}$ is known and we compute the counts on facilities as
%$$
%\hat q_e^t \gets \pst(\omega_e^{\leq t}, \hat q_e^{t-1}, \epsilon') \qquad \hat y_e^t \gets \abm(\hat q_e^t,t) \qquad \forall t \in [T],
%$$
%then the vector $(\hat y_e^t)_{t \in [T]}$ has the same distribution as the output of $\bm(\omega_e^{\leq T}, \epsilon')$.

It will be necessary to control the error, bounded in Lemma \ref{lem:chanerror} for a single counter, in the worst case over the counters on all facilities in the congestion game.  We will then use $\beta' = \beta/m$ and write:
\begin{equation}
\err= \frac{\sqrt{8\log(T)\log(2m/\beta)}}{\epsilon'}.
\label{eq:error}
\end{equation}
We can then bound the difference $\mathbf{\Delta}$, with probability $1-\beta$, between the noisy costs and the actual costs summed over all facilities being used by a player for a given action profile $\ba^t$
\begin{align}
| \hat c(\tau_i,\ba^t) - c(\tau_i,\ba^t) |& \leq \sum_{e\in a_i^t}\left|\ell_e(\hat y_e(\ba^t)) - \ell_e(y_e(\ba^t)) \right|.\nonumber \\
& \leq  (m \sigma_\ell) \err = \mathbf{\Delta}.
\label{eq:costdiff}
\end{align}
We make the following low error assumption (which is satisfied with probability at least $1-\beta$ by Lemma \ref{lem:chanerror}) so that we can bound the error from the noisy counts with high probability.
\begin{assumpt}(Low Error Assumption)
For each $t \in [T]$ and each facility stream $\{\omega_e^{\leq t}: e \in E\}$, we assume that the error introduced by the adaptive binary mechanism $\abm$ is at most $\err$ (given in \eqref{eq:error})
\end{assumpt}
We can bound the maximum number of $\alpha$-noisy best responses  that can be made under noisy best response dynamics before play converges to an approximate Nash equilibrium. This follows the classic potential function method of \citeasnoun{MS96}.
\begin{lemma}\label{lem:Tnoise}
Under the low error assumption, and for $\alpha > 4\mathbf{\Delta}$, the total number of  $\alpha$-noisy best response moves that can be made before no player has any remaining $\alpha$-noisy best response moves (i.e. play converges to an approximate equilibrium) is bounded, which allows us to set $T$ to be:
\begin{equation}\label{eq:Tnoise}
T = \frac{2mn}{\alpha}.
\end{equation}
\end{lemma}
\begin{definition}
(Improvement) The \emph{improvement} of player $i$ of type $\tau_i$ at a given action profile $\ba$ is defined to be the decrease in cost he would experience by making a best response (according to his exact cost function):
$$
Improve(\tau_i, \ba)= c(\tau_i,\ba) -\min_{a_i'\in A_i}\{ c(\tau_i, (a_i', a_{-i})) \}.
$$
\end{definition}
%We have the following claim which coincides with the claim made when we analyze  {\bf BR-EXACT} in the appendix.
\begin{claim}\label{claim:one}
Let $\ba^t$ be some action profile at time $t \in [T]$.  If player $i$ makes an $\alpha$-noisy best response at time $t$, then under the low error assumption:
\begin{equation*}
Improve(\tau_i, \ba^t)  \leq 2\mathbf{\Delta}
\end{equation*}
\end{claim}
Note that if player $i$ of type $\tau_i$ does not make an $\alpha$-best response at time $t$, then we can bound her improvement by
$$
Improve(\tau_i, \ba^{t+1}) \leq Improve(\tau_i, \ba^t) + m\sigma_\ell.
$$
We wish to lower bound the \emph{gap} between pairs of time steps $t, t'$ at which any player $i$ can make $\alpha$-noisy best responses.
\begin{definition}(Noisy Gap)  The \emph{noisy gap} $\Gamma$ of a congestion game $\G$ is:
\begin{align*}
\Gamma = \min\{ |t-t'| : t \neq t', \text{some player $i$ makes an $\alpha$-noisy best responses at $t$ and $t'$} \}.
\end{align*}
\end{definition}
\begin{lemma}
\label{lem:boundedmoves}
Assume that $\alpha > 4 \mathbf{\Delta}$. Under the low error assumption, the noisy gap $\Gamma$ of a congestion game $\G$ satisfies
$$
\Gamma \geq \frac{\alpha - 2\mathbf{\Delta} }{m \sigma_\ell}.
$$
Furthermore, the number $p$ of $\alpha$-noisy best responses any individual player can make is bounded by
\begin{equation}
p \leq \frac{4m^2n\sigma_\ell}{\alpha^2}.
\label{eq:knoise}
\end{equation}
\label{lem:noisygap}
\end{lemma}
%\begin{proof}
%Let $t$ and $t+t'$ be two times that player $i$ of type $\tau_i$ made $\alpha$-noisy best responses.  We have
%$$
%Regret(\tau_i, r^{t+t'}) \leq 2\mathbf{\Delta}_T(\beta) + (t') m\Delta\ell.
%$$
%We set $t' = \frac{\alpha -2 \mathbf{\Delta}_T(\beta) }{m\Delta\ell}$.  Further, we know from \eqref{eq:Tnoise} that under the low error assumption, there can be as many as $\frac{2mn}{\alpha}$ total $\alpha$-noisy best responses.  With our bound on the noisy gap we obtain a bound on the number of $\alpha$-noisy best response moves $p$ a single player can make.
%$$
%p \leq \frac{\frac{2mn}{\alpha}}{\hat\gamma} \leq \frac{2m^2n\Delta\ell}{\alpha(\alpha-2\mathbf{\Delta}_T(\beta))}
%$$
%and for $\alpha > 4 \mathbf{\Delta}_T(\beta)$ we obtain the stated result for $p$.
%\end{proof}
%%%%%   INTRODUCE THE ALGORITHM %%%%%%%%%
%%%%%%%%%%%%%%%%%%%%%%%%%%%%%%%%%%%%%%%%%%%%%%%%%%%%%
\begin{algorithm}
\caption{Private Best Responses}\label{BR-PRIVATE}
\begin{algorithmic}[0]
\INPUT : For a sufficiently large congestion game, users report type vector $\tau$.
\OUTPUT : Suggested action profile $\ba$
\Procedure {$\pbr$} {$\tau$}
\State Set: $\alpha =\Theta\left( \left(\frac{m^4n (\sigma_\ell)^2 \log^{2}(mn/\beta)}{\epsilon}\right)^{1/3}\right)$,
$T = \frac{2mn}{\alpha}$, $p = \frac{4m^2n\sigma_\ell}{\alpha^2}$, and $\epsilon ' = \frac{\epsilon}{\ifaddone 3 \else 2 \fi pm\log(T)}$
\For{$i \in [n]$}
	\State $count(i) \gets 1$
	\State Choose $a_i \in A_i$ to be some arbitrary initial action.
	\For{$e \in E$}
		\If{$e \in a_i^i$} $w_e^i \gets 1$, {\bf  else } $\quad w_e^i \gets 0$
		\EndIf
	\EndFor
\EndFor
	\For{$e \in E$}
	$\quad \hat \rr^1_e \gets 0^{\ifaddone \left( \lfloor\log(nT)\rfloor+1\right)\else \lfloor\log(nT)\rfloor \fi  \times nT}$
	 	\For{$j = 2, \cdots, n$ }
			 $\quad \hat \rr_e^{j} \gets \pst(\omega_e^{\leq j},\hat\rr_e^{j-1}, \epsilon')$
		\EndFor
		 \State $\hat y_e^n \gets \abm(\hat \rr_e^n, n)$
	 \EndFor
$\ba^n = (a_i)_{i = 1}^n$
\For{$t = n+1,\cdots, nT$} $ i \gets \text{Player at round $t$} $
		\State $ a_i^t \gets a_i^*$, an $\alpha$-Best Response for $i$ given the past facility count $\hat y_e^{t-1}$
		\If{$a_i^t == NA$} $\quad a_i^t \gets a_i^{t-1}$
			\For{$e \in E$}
				 $\omega_e^t \gets 0$
			\EndFor
		\Else			
		\State $count(i) \gets count(i)+1$
			\For{$e \in E$} $ \omega_e^t \gets 0$
				\If{$e\in a_i^{t-1}\backslash a_i^t$}
					 $\omega_e^t \gets (-1)$,
				{\bf else if } $e\in a_i^{t} \backslash a_i^{t-1}$
					{\bf then } $ \omega_e^t \gets 1$.
				\EndIf
			\EndFor
		\EndIf
		\If{$count(i) > p$}
			 HALT, \Return FAIL
		\EndIf
		\For{$e \in E$}
		 $\quad \hat{\rr}_e^t \gets \pst(\omega_e^{\leq t}, \hat{\rr}_e^{t-1} , \epsilon'), \qquad \hat y_e^t \gets \abm(\hat \rr_e^t, t)$
	\EndFor
	\For{$j \neq i$}
		 $\quad a_j^t \gets a_j^{t-1}$
	\EndFor
\EndFor \\
\Return $a_i^{nT}$ to player of type $\tau_i$ for every $i \in [n]$
\EndProcedure
\end{algorithmic}
\end{algorithm}
%%%%%%%%%%%%%%%%%%%%%%%%%%%%%%%%%%%%%%%%%%%%%%%%%%%%
\indent  Given the players' input type vector $\tau$, $\pbr(\tau)$ in Algorithm \ref{BR-PRIVATE} simulates best response dynamics with only access to the noisy counts of players on each facility at each iteration.  The algorithm cycles through the players in a round robin fashion and in order, each player $i$ makes an $\alpha$-noisy best response if she has one. If not, she continues playing her current action (we denote her noisy best response in this case as $NA$). After each move, the algorithm updates the counts on each facility via $\abm$ a modified version of the binary mechanism to deal with adaptive streams (see Algorithm \ref{ABM}) and computes the updated noisy costs for the next iteration.  The algorithm terminates after it has iterated $nT$ times (which is enough iterations for $T$ noisy best responses, which under the low error assumption guarantees that we have found an equilibrium), or if any individual player has changed her action more than $p$ times (which we have shown can only occur if the low error assumption is violated, which is an event with probability at most $\beta$).  If the former occurs, the algorithm outputs the final action profile, and in the latter case, the algorithm outputs that it failed.

\subsection{Analysis of $\pbr$}
The analysis of the above algorithm gives the following privacy guarantee.

\begin{theorem}
The algorithm  $\pbr(\tau)$ is $(\epsilon,\beta)$-jointly differentially private.
%\begin{equation}
%\epsilon ' = \frac{\epsilon}{2pm\log(T)}.
%\label{eq:eplem}
%\end{equation}
\label{lem:dp}
\end{theorem}

%%%%%% PROOF SKETCH %%%%%%%
We give a proof sketch here to highlight the main ideas of the formal proof given in the appendix.  We observe that the action suggestion output to each player can be computed as a function of his own type $\tau_i$, and the output history of the noisy counters. This is because each player $i$ can compute which action he would have taken (if any) at every time step that he is slotted to move by observing the noisy counts: his final action is simply the last noisy best response he took over the course of the algorithm. It suffices to prove that the noisy partial sums from which the counts are computed (internally in the binary mechanism) are differentially private.

\citet{Binary} prove that a single counter facing a non-adaptively chosen (sensitivity 1) stream of bits is differentially private. We need to generalize this result to multiple counters facing adaptively chosen streams, whose joint sensitivity (i.e. the number of times any player ever changes his action) is bounded. This is exactly the quantity that we bounded in Lemma \ref{lem:boundedmoves}. We want the privacy parameter to scale like the joint sensitivity of the streams, rather than the sum of the sensitivities of each of the individual streams. This is akin to the (standard) analysis of privately answering a vector-valued query with noise that scales with its $\ell_1$ sensitivity, rather than with the sum of its coordinate-wise sensitivities.

Finally, we achieve $(\epsilon,\beta)$-joint differential privacy rather than $(\epsilon,0)$-joint differential privacy because the publicly observable event that the algorithm fails without outputting an equilibrium might be disclosive. However, this event occurs with probability at most $\beta$.
%%%%%%%%%%%%%%%%%%%%%%%

This theorem satisfies the first of the two necessary hypotheses of Theorem \ref{lem:main}. We now show that with high probability, our algorithm outputs an approximate Nash equilibrium of the underlying game, which satisfies the remaining hypothesis of Theorem \ref{lem:main}.  We need to show that when $\pbr(\tau)$ outputs an action profile $\ba$, players have little incentive to deviate from using their suggested action, leaving details to the appendix.
\begin{theorem}
With probability $1-\beta$, the algorithm $\pbr(\tau)$
%\begin{equation}
%\alpha =\Theta\left( \left(\frac{m^4n\log^{3/2}(n)(\sigma_\ell)^2\sqrt{\log(1/\beta)}}{\epsilon}\right)^{1/3}\right)
%\label{eq:alpha}
%\end{equation}
 produces an $\eta$-approximate Nash Equilibrium, where
\begin{align}
%\eta =  O\left( \left(\frac{m^4n\log^{3/2}(n)(\sigma_\ell)^2\log(1/\beta)}{\epsilon}\right)^{1/3}\right)
\eta =  \tilde O\left( \left(\frac{m^4n (\sigma_\ell)^2 \log^2(1/\beta)}{\epsilon}\right)^{1/3}\right) =  \tilde O\left( \left(\frac{m^2n \lambda^2 \log^2(1/\beta) }{\epsilon}\right)^{1/3}\right)
\label{eq:etanash}
\end{align}
\label{thm:nash_eq}
\end{theorem}
We have satisfied the hypotheses of Theorem \ref{lem:main} (with the slight notational difference that we are dealing with costs instead of utilities) because we have an algorithm $\M(\tau)  = \pbr(\tau)$ that is $(\epsilon, \beta)$ joint differentially private and with probability $1-\beta$ produces an action profile that is an $\eta$ approximate pure strategy Nash Equilibrium of the congestion game defined by the type vector $\tau$.  Note that the upper bound $U$ on utilities needed in Theorem \ref{lem:main} is here $U = m$.  In the following theorem we choose the parameter values for $\epsilon$ and $\beta$ that optimize for the approximation parameter for $\eta'$ in Theorem \ref{lem:main} given $\eta$ in \eqref{eq:etanash}

\begin{theorem}
\label{thm:weak_main}
Choosing $\pbr(\tau)$ as our mediator $\M(\tau)$ for the congestion game $\G$, for any $\tau \in \mathcal{T}^n$  it is, with probability at least $1-o(1)$, an $\eta'$-approximate ex post Nash Equilibrium of the weakly mediated game $\G_\M^W$ for each player to play the ``good" action profile $(\tau,\mathbf{f})$ where $\eta'$ is
\begin{align}
\eta'  = \tilde O \left( \left(m^5n (\sigma_\ell)^2 \log^2(1/\sigma_\ell) \right)^{1/4}\right) = \tilde O \left( m^{3/4} n^{1/4}  \sqrt{\lambda \log(1/\lambda)} \right)
\label{eq:etaprime}
\end{align}
and $\lambda = m \sigma_\ell$ is the largeness parameter for $\cG$.
\end{theorem}

Note that for any value of the sensitivity that shrinks with $n$ (treating $m$ as constant) at a rate of $\sigma_\ell = o\left(\frac{1}{\sqrt{n}\log(n)}\right)$, our approximation error $\eta$ asymptotes to $0$ as $n$ grows large.

%%%%%%%%%%%%%%%%%%%%%%%%%%%%%%%%%%%%%%%%%%%%%%%%%%%%%%%%%%%%%%%%%%%%%%%%%%%%%%
%%%%%%%%%%%%%%%%%%%%%%%%%%%%%%%%%%%%%%%%%%%%%%%%%%%%%%%%%%%%%%%%%%%%%%%%%%%%%%
\section{Strong Mediators for General Games} \label{sec:ub} \label{sec:strongmediatorlaplace}
%%%%%%%%%%%%%%%%%%%%%%%%%%%%%%%%%%%%%%%%%%%%%%%%%%%%%%%%%%%%%%%%%%%%%%%%%%%%%%
%%%%%%%%%%%%%%%%%%%%%%%%%%%%%%%%%%%%%%%%%%%%%%%%%%%%%%%%%%%%%%%%%%%%%%%%%%%%%%

In the previous section, we showed that there is a jointly differentially private algorithm to compute a Nash equilibrium in a congestion game (and thus a weak mediator for congestion games).  We do not know if such an algorithm exists for arbitrary large games (If such an algorithm does exist, it cannot be computationally efficient, since it is $PPAD$-hard to compute Nash equilibria in arbitrary games, even ignoring the privacy constraint).
However, in the following we  show that we can privately compute a \emph{correlated equilibrium} of arbitrary large games, and hence can obtain a strong mediator that makes nice behavior an ex-post Nash equilibrium, in arbitrary large games.

The high level idea of the construction is similar to the previous section, except that we now simulate no-regret learning dynamics, rather than best response dynamics. No-regret dynamics are guaranteed to converge to correlated equilibria in any game (see, e.g., \citeasnoun{hart2000simple}). As in the previous section, instead of using the actual realized regret, we use a noisy version to ensure privacy. We show that no-regret dynamics are robust to the addition of noise, and that since the game is large, it suffices to add ``small'' noise in order to guarantee joint differential privacy.

%%%%%%%%%%%%%%%%%%%%%%%%%%%%%%%%%%%%%%%%%%%%%%%%%%%%%%%%%%%%%%%%%%%%%%%%%%%%%%
%%%%%%%%%%%%%%%%%%%%%%%%%%%%%%%%%%%%%%%%%%%%%%%%%%%%%%%%%%%%%%%%%%%%%%%%%%%%%%

%%%%%%%%%%%%%%%%%%%%%%%%%%%%%%%%%%%%%%%%%%%%%%%%%%%%%%%%%%%%%%%%%%%%%%%%%%%%%%
\subsection{Definitions and Basic Properties}
%%%%%%%%%%%%%%%%%%%%%%%%%%%%%%%%%%%%%%%%%%%%%%%%%%%%%%%%%%%%%%%%%%%%%%%%%%%%%%
  First, we recall some basic results on no-regret learning. See \citeasnoun{nisan2007algorithmic} for a text-book exposition.

Let $\actionset$ be a finite set of $k$ \emph{actions}.  Let $\Losses = (\losses_1, \dots, \losses_{T}) \in [0,1]^{k \times T}$ be a \emph{loss matrix} consisting of $T$ vectors of losses for each of the $k$ actions.  Let $\Pi = \set{\pi \in [0,1]^{k} \mid \sum_{j = 1}^{k} \pi^j = 1}$ be the set of distributions over the $k$ actions and let $\pi_{U}$ be the uniform distribution.  An \emph{online learning algorithm} $\nralg \from \Pi \times [0,1]^{k} \to \Pi$ takes a distribution over $k$ actions and a vector of $k$ losses, and produces a new distribution over the $k$ actions.  We use $\nralg_{t}(\Losses)$ to denote the distribution produced by running $\nralg$ sequentially $t-1$ times using the loss vectors $\losses_1, \dots, \losses_{t-1}$, and then running $\nralg$ on the resulting distribution and the loss vector $l_t$.  That is:
\begin{align*}
&\nralg_{0}(\Losses) = \pi_{U},\\
& \nralg_{t}(\Losses) = \nralg(\nralg_{t-1}(\Losses), \tlosses).
\end{align*}
We use $\vec{\nralg}(\Losses) = (\nralg_0(\Losses), \nralg_1(\Losses), \dots, \nralg_T(\Losses))$ to be the vector of distributions when $T$ is clear from context.

Let $\state_0, \dots, \state_{T} \in \Pi$ be a sequence of $T$ distributions and let $\Losses$ be a $T$-row loss matrix. We define the quantities:
\begin{align*}
&\exploss(\state, \losses) = \sum_{j=1}^{k} \pi^j l^j,  \\
&\exploss( \state_0, \dots, \state_{T}, \Losses) = \frac{1}{T} \sum_{\round = 1}^{\rounds} \exploss(\state_{t}, \losses_{t}) , \\
&\exploss(\vec{\nralg}(\Losses), \Losses') = \exploss(\nralg_{0}(\Losses), \nralg_{1}(\Losses), \dots, \nralg_{T}(\Losses), \Losses')
\\
& \exploss(\vec{\nralg}, \Losses) = \exploss(\vec{\nralg}(\Losses), \Losses).
\end{align*}
Note that the notation retains the flexibility to run the algorithm $\nralg$ on one loss matrix, but measure the loss $\nralg$ incurs on a different loss matrix. This flexibility will be useful later.

Let $\mods$ be a family of functions $\mod \from \actionset \to \actionset$.  For a function $\mod$ and a distribution $\state$, we define the distribution $\modstate$ to be
\begin{equation*}
(\modstate)^j = \sum_{j': f(j') = j} \state^{j'}.
\end{equation*}
The distribution $\modstate$ corresponds to the distribution on actions obtained by first choosing an action according to $\pi$, then applying the function $f$.

Now we define the following quantities:
\begin{align*}
&\exploss(\state_1, \dots, \state_T, \Losses, \mod) = \exploss(\modstate_1, \modstate_2, \dots, \modstate_T, \Losses), \\
&\regret(\vec{\nralg}(\Losses), \Losses', \mod) =  \exploss(\vec{\nralg}(\Losses), \Losses') - \exploss(\vec{\nralg}(\Losses), \Losses', \mod),\\
&\regret(\vec{\nralg}, \Losses, \mod) = \regret(\vec{\nralg}(\Losses), \Losses, \mod) , \\
&\regret(\vec{\nralg}, \Losses, \mods) = \max_{\mod \in \mods} \regret(\vec{\nralg}, \Losses, \mod).
\end{align*}
As a mnemonic, we offer the following.  $\Lambda$ refers to expected loss, $\rho$ refers to regret. Next, we define the families $\fixedmods, \swapmods:$
\begin{align*}
&\fixedmods = \set{\mod_j(j') = j, \textrm{ for all } j' \mid j \in \actionset} \\
&\swapmods = \set{f: \actionset \to \actionset}
\end{align*}

Looking ahead, we will need to be able to handle not just \emph{a priori fixed} sequences of losses, but also adaptive sequences. To see why, note that for a game setting, a player's loss will depend on the distribution of actions played by everyone in that period, which will depend, in turn, on the losses everyone experienced in the previous period and how everyone's algorithms reacted to that.
\begin{definition}[Adapted Loss] \label{def:adaptloss}
A loss function $\cL$ is said to be adapted to an algorithm $\cA$ if in each period $t$, the experienced losses $l_t \in [0,1]^k$ can be written as:
\begin{align*}
l_t = \cL(l_0, \cA_0(l_0), l_1, \cA_1(l_1),\ldots, l_{t-1}, \cA_{t-1}(l_{t-1})).
\end{align*}
\end{definition}
The following well-known result shows the existence of algorithms that guarantee low regret even against adapted losses (see e.g. \citeasnoun{nisan2007algorithmic}).

\begin{theorem} \label{thm:noregretalgsexist}
There exists an algorithm $\fixedalg$ such that for any adapted loss $\cL$, $\regret\left(\overrightarrow{\fixedalg}, \cL,  \fixedmods\right) \leq \sqrt{\frac{2\log k}{T}}$.  There also exists an algorithm $\swapalg$ such that $\regret\left(\overrightarrow{\swapalg}, \cL, \swapmods\right) \leq k\sqrt{\frac{2 \log k}{T}}$.
\end{theorem}

%%%%%%%%%%%%%%%%%%%%%%%%%%%%%%%%%%%%%%%%%%%%%%%%%%%%%%%%%%%%%%%%%%%%%%%%%%%%%%
\subsection{Noise Tolerance of No-Regret Algorithms}\label{sec:no-regret}
%%%%%%%%%%%%%%%%%%%%%%%%%%%%%%%%%%%%%%%%%%%%%%%%%%%%%%%%%%%%%%%%%%%%%%%%%%%%%%
The next lemma states that when a no-regret algorithm is run on a noisy sequence of losses, it does not incur too much additional regret with respect to the real losses.
\begin{lemma} [Regret Bounds for Bounded Noise] \label{lem:noisyregretsbounded}
Let $\Losses \in [\frac{1}{3}, \frac{2}{3}]^{T \times k}$ be any loss matrix, $\Noises = (\tjnoise) \in [-\zeta, \zeta]^{T \times k}$ be an arbitrary matrix with bounded entries, and let $\noisyLosses = \Losses + \Noises$.  Further, we have an algorithm $\nralg$ and a family of functions $\mods$.  Then
\begin{equation*}
\regret(\vec{\nralg}(\noisyLosses), \Losses, \mods) \leq \regret(\vec{\nralg}, \noisyLosses, \mods) + 2\zeta.
\end{equation*}
\end{lemma}

\begin{corollary}\label{thm:lowregretbounded}
Let $\Losses \in [\frac{1}{3}, \frac{2}{3}]^{T \times k}$ be any loss matrix and let $\Noises \in \R^{T \times k}$ be a random matrix such that $\Prob{\Noises}{\Noises \in [- \zeta, \zeta]^{T \times k}} \geq 1-\beta$ for some $\zeta \in [0,\frac{1}{3}]$, and let $\noisyLosses = \Losses + \Noises$.  Then
\begin{enumerate}
\item $\Prob{\Noises}{\regret\left(\overrightarrow{\fixedalg}(\noisyLosses), \Losses, \fixedmods\right) > \sqrt{\frac{2 \log k}{T}} + 2\zeta} \leq \beta,$
\item $\Prob{\Noises}{\regret\left(\overrightarrow{\swapalg}(\noisyLosses), \Losses, \swapmods\right) > k \sqrt{\frac{2 \log k}{T}} + 2\zeta} \leq \beta.$
\end{enumerate}
\end{corollary}

Note that the technical condition $\zeta \in [0,\frac{1}{3}]$ is needed to ensure that the noisy loss matrix $\noisyLosses$ is contained in $[0,1]^{T \times k}$, which is required to apply the regret bounds of Theorem~\ref{thm:noregretalgsexist}.

Next, we state a tighter bound on the additional regret in the case where the entries of $Z$ are i.i.d. samples from a Laplace distribution using Theorem \ref{thm:conc}.
\begin{lemma} [Regret Bounds for Laplace Noise] \label{lem:noisylapregretsbounded}
Let $\Losses \in [\frac{1}{3}, \frac{2}{3}]^{T \times k}$ be any loss matrix. Let $\Noises = (\tjnoise) \in \R^{T \times k}$ be a random matrix formed by taking each entry to be an independent sample from $\Lap(b)$, and let $\noisyLosses = \Losses + \Noises$.  Let $\nralg$ be an algorithm.  Let $\mods$ be any family of functions. Then for any $\alpha \leq b$.
\begin{equation*}
\Prob{\Noises}{ \regret(\vec{\nralg}(\noisyLosses), \Losses, \mods) - \regret(\vec{\nralg}, \noisyLosses, \mods)  > \alpha} \leq 2|\mods| e^{- \alpha^2 T/ 24b^2}.
\end{equation*}
\end{lemma}

\begin{corollary}\label{thm:lowregretlaplace}
Let $\Losses \in [\frac{1}{3}, \frac{2}{3}]^{T \times k}$ be any loss matrix and let $\Noises \in \R^{T \times k}$ be a random matrix formed by taking each entry to be an independent sample from $\Lap(b)$ for $b < \frac{1}{6 \log(4kT/\beta)}$ and let $\noisyLosses = \Losses + \Noises$.  Then
\begin{enumerate}
\item $\Prob{\Noises}{\regret\left(\overrightarrow{\fixedalg}(\noisyLosses), \Losses, \fixedmods\right) > \sqrt{\frac{2 \log k}{T}} + b\sqrt{\frac{24\log(4k/\beta)}{T}}} \leq \beta$,
\item $\Prob{\Noises}{\regret\left(\overrightarrow{\swapalg}(\noisyLosses), \Losses, \swapmods\right) > k \sqrt{\frac{2 \log k}{T}} + b\sqrt{\frac{24k\log(4k/\beta)}{T}} } \leq \beta$.
\end{enumerate}
\end{corollary}

As before, the technical condition upper bounding $\zeta$ is to ensure that the noisy loss matrix $\noisyLosses$ is contained in $[0,1]^{T \times k}$, so that the regret bounds of $\cA$ apply.

\subsection{From No Regret to Equilibrium} \label{sec:noregrettoeq}

Let $(\tau_{1}, \dots, \tau_{n})$ be the types for each of $n$ players.  Let $\mathcal{C} = \set{(\state_{i,1}, \dots, \state_{i,T})}_{i=1}^n$ be a collection of $n$ sequences of distributions over $k$ actions, one for each player.  Let $\set{(\losses_{i,1}, \dots, \losses_{i,T})}_{i =1}^n $ be a collection of $n$ sequences of loss vectors $\losses_{i,t} \in [0,1]^{k}$ formed by the action distribution.  More formally, for every $j$, $\loss^{j}_{i,t}  = 1 - \Ex{\ba_{-i}\sim\state_{-i,t}}{u(\tau_i,(j, \ba_{-i}))} $ for nonnegative utilities bounded by 1.
Define the maximum regret that any player has to her losses
\begin{equation*}
\mregret(\mathcal{C}, L, \mods) = \max_{i} \regret(\mathcal{C}_i, L_i, \mods)
\end{equation*}
where $\mathcal{C}_{i} = (\state_{i,0}, \dots, \state_{i,T})$ and $L_i = (\losses_{i,1}, \dots, \losses_{i,T})$.

Given the collection $\mathcal{C}$, we define the correlated action distribution $\Pi_{\mathcal{C}}$ to be the average distribution of play.  That is, $\Pi_{\mathcal{C}}$ is the distribution over $A$ defined by the following sampling procedure: Choose $t$ uniformly at random from $\set{1,2,\dots,T}$, then, for each player $i$, choose $a_i$ randomly according to the distribution $\state_{i,t}$, independently of the other players.

The following well known theorem (see, e.g. \citeasnoun{nisan2007algorithmic}) relates low-regret sequences of play to the equilibrium concepts  (Definition~\ref{def:ace}):
\begin{theorem} \label{thm:regtoeq}

If the maximum regret with respect to $\fixedmods$ is small, i.e.  $\mregret(\mathcal{C},L,\fixedmods) \leq \eta$ and players utilities are nonnegative bounded by $1$ then the correlated action distribution $\Pi_{\mathcal{C}}$ is an $\eta$-approximate coarse correlated equilibrium.  Similarly, if $\mregret(\mathcal{C},L,\swapmods) \leq \eta$, then $\Pi_{\mathcal{C}}$ is an $ \eta$-approximate correlated equilibrium.
\end{theorem}

In this section we show that no-regret algorithms are noise-tolerant, that is we still get good regret bounds with respect to the real losses if we run a no-regret algorithm on noisy losses (real losses plus low-magnitude noise).

Let $\Losses \in [0,1]^{T \times k}$ be a loss matrix.  Define $\scaledLosses = \frac{\Losses + 1}{3}$ (entrywise) and note that $\scaledLosses \in [\frac{1}{3}, \frac{2}{3}]^{T \times k}$.  The following lemma states that running $\nralg$ on $\scaledLosses$ doesn't significantly increase the regret with respect to the real losses.
\begin{lemma} \label{lem:lossscaled}
For every algorithm $\nralg$, every family $\mods$, and every loss matrix $\Losses \in [0,1]^{T \times k}$,
\begin{equation*}
\regret(\vec{\nralg}(\scaledLosses), \Losses, \mods) \leq 3 \regret(\vec{\nralg}, \scaledLosses, \mods).
\end{equation*}
In particular, for every $\Losses \in [0,1]^{T \times k}$
\begin{equation*}
\regret\left(\overrightarrow{\fixedalg}(\scaledLosses), \Losses, \fixedmods\right) \leq \sqrt{\frac{18\log k}{T}}
\;\text{ and }\;
\regret\left(\overrightarrow{\swapalg}(\scaledLosses), \Losses, \swapmods\right) \leq k \sqrt{\frac{18\log k}{T}}.
\end{equation*}
\end{lemma}

In light of Lemma \ref{lem:lossscaled}, for the rest of this section we will take $\Losses$ to be a loss matrix in $[\frac{1}{3}, \frac{2}{3}]^{T \times k}$.  This rescaling will only incur an additional factor of $3$ in the regret bounds we prove.  Let $\Noises \in \R^{T \times k}$ be a real valued \emph{noise matrix}.  Let $\noisyLosses = \Losses + \Noises$ (entrywise).  In the next section we  consider the case where $Z$ is an arbitrary matrix with bounded entries.  We prove a tighter bound for the case where $Z$ consists of independent draws from a Laplace distribution.

Having demonstrated the noise tolerance of no-regret algorithms, we now argue that for appropriately chosen noise, the output of the algorithm constitutes a jointly-differentially private mechanism (Definition \ref{def:privacy}).
%We prove two results of this type. First, in Section \ref{sec:fewactions} we consider games with `few' actions per player. While the algorithm is conceptually more straightforward, it is not useful in certain cases of interest.  For example, in the routing games we considered in the previous section with weak mediators, the set of actions available to a player consists of all routes between her starting point and her destination.  Even if the graph (road network) is small, the number of feasible routes can be extremely large (exponential in the number of edges (roads)). However, in such games, the set of types (utility functions) is small (i.e.~the set of all source-destination pairs). Motivated by this observation, in Appendix \ref{sec:smalltypes}  we consider games with large action spaces, but bounded type spaces.
At a high-level, our proof has two main steps. First, we construct a \emph{wrapper} $\nrl^{\nralg}$ that will ensure privacy.  The wrapper takes as input the parameters of the game, the reported tuple of types, and any no-regret algorithm $\nralg$.  This wrapper will attempt to compute an equilibrium using the method outlined in Section~\ref{sec:noregrettoeq}---for $T$ periods, it will compute a mixed strategy for each player by running $\nralg$ on the previous period's losses.  In order to ensure privacy, instead of using the true losses as input to $\nralg$, it will use losses perturbed by suitably chosen Laplace noise.  After running for $T$ periods, the wrapper will output to each player the sequence of $T$ mixed strategies computed for that player.  In Theorem \ref{thm:privatecnrl} we show that this constitutes a jointly differentially private mechanism. Then, in Theorem \ref{thm:accCElaplace}, we show that the output of this wrapper converges to an approximate correlated equilibrium when the input algorithm is the no-swap-regret algorithm $\swapalg$.

\subsubsection{Noisy No-Regret Algorithms are Differentially Private}\label{sec:noisyDP}

\begin{algorithm}
\caption{``Wrapper" Algorithm}
\label{NR-Laplace}
\begin{algorithmic}[0]
\INPUT : Vector of types.
\OUTPUT : $T$ different distributions for each player over the $k$ actions.
\Procedure {$\nrl^{\nralg}$}{$\tau_{1}, \dots \tau_{n}$}
%$\textsc{NRLaplace}^{\nralg}(u_{1}, \dots u_{n})$
%\begin{algorithmic}
\State Parameters: $\eps, \delta, \lambda \in (0,1], n, k, T \in \N$
\State Let: $\state_{1,1}, \dots, \state_{n,1}$ each be the uniform distribution over $\actionset$.
\State Let: $b = \lambda \eps^{-1} \sqrt{8nkT \log(1/\delta)}$
\For{$t = 1, 2, \dots, T$}
	\State Let: $\itjloss = 1 - \Ex{\state_{-i,t}}{u(\tau_i,(j, a_{-i}))}$ for every player $i$, action $j$.
	\State Let:  $\itjnoise$ be an i.i.d. draw from $\Lap(b)$ for every player $i$, action $j$.
	\State Let: $\noisyitjloss = \itjloss + \itjnoise$ for every player $i$, action $j$.
	\State Let: $\state_{i,t+1} = \nralg(\itstate, \noisyitlosses)$ for every player $i$.
\EndFor
\Return $(\state_{i,1}, \dots, \state_{i,T})$ to player $i$, for every $i$.
\EndProcedure
\end{algorithmic}
\end{algorithm}

\begin{theorem}[Privacy of $\nrl^\nralg$] \label{thm:privatecnrl}
For any $\nralg$, the algorithm $\nrl^{\nralg}$ satisfies $(\eps, \delta)$-joint differential privacy.
\end{theorem}

We now sketch the proof. Fix any player $i$ and any tuple of types $\tau_{-i} \in \mathcal{T}^{n-1}$. We argue that the output to all other players is differentially private as a function of $\tau_{i}$.   It will be easier to analyze a modified mechanism that outputs the noisy losses $(\noisylosses_{-i,1}, \dots, \noisylosses_{-i,T})$, rather than the mixed strategies $(\state_{-i,1},\dots,\state_{-i,T})$.  Since the noisy losses are sufficient to compute $(\state_{-i,1}, \dots, \state_{-i,T})$, proving that the noisy losses are jointly differentially private is sufficient to prove that the mixed strategies are as well.

To get intuition for the proof, first consider the first period of noisy losses $\noisylosses_{-i,1}$.  For each player $i' \neq i$, and each action $j \in [k]$, the loss $l^{1}_{i',j}$ depends on $\pi_{i,1}$, which is independent of the utility of player $i$.  Thus, in the first round there is no loss of privacy.  In the second round, the loss $l^{2}_{i',j}$ depends on $\pi_{i,2}$, which depends on the losses for player $i$ in period $1$, and thus depends on the utility of player $i$.  The loss $l^{2}_{i',j}$ also depends on the mixed strategies $\pi_{i'',2}$ for players $i'' \neq i, i'$, but as we have argued these mixed strategies are independent of $\tau_{i}$.   We will take a pessimistic view and assume that changing player $i$'s type from $\tau_i$ to $\tau'_{i}$ will change $\pi_{i,2}$ \emph{arbitrarily}.  The assumption that the underlying game is $\lambda$-large ensures that the expected losses of player $i'$, $l^{2}_{i',j}$ only change by at most $\lambda$.  Thus, by Theorem~\ref{thm:laplaceprivacy} and our choice of the noise parameter $b$, each noisy loss $\widehat{l}^{2}_{i',j}$ will be $\eps / \sqrt{8nkT \log(1/\delta)}$ differentially private as a function of $\tau_{i}$.

Understanding the third round will be sufficient to argue the general case.  Just as in period $2$, the loss $l^{3}_{i',j}$ depends on $\pi_{i,3}$.  However, the $l^{3}_{i',j}$ also depends on $\pi_{i'',3}$ for players $i'' \neq i,i'$ and now these strategies do indeed depend on the utility of $i$ (and thus $\tau_i$), as we saw when reasoning about period $2$.  However, the key observation is that $\pi_{i'',3}$ depends on the utility of player $i$ only through the noisy losses $l^{2}_{i'',j}$ that we computed in the previous round.  Since we already argued that these losses are differentially private as a function of $\tau_{i}$, it will not compromise privacy further to use these noisy losses when computing $l^{3}_{i',j}$.  Thus, conditioned on the noisy losses output in periods $1$ and $2$, the losses $l^{3}_{i',j}$ depend only on the mixed strategy of player $i$ in period $3$.  As we argued before, the amount of noise we add to these losses will be sufficient to ensure $\eps / \sqrt{8nkT \log(1/\delta)}$ differential privacy as a function of $\tau_{i}$.

In summary, we have shown that for every period $t$, every player $i' \neq i$, and every action $j$, the noisy loss $\widehat{l}^{t}_{i',j}$ is an $\eps/\sqrt{8nkT \log(1/\delta)}$-differentially private function of $\tau_{i}$ and of the previous $t-1$ periods' noisy losses, which are themselves already differentially private.  In total we compute $T(n-1)k$ noisy losses.  Hence, the adaptive composition theorem (Theorem~\ref{thm:advcomp} in the appendix) ensures that the entire sequence of noisy losses $\widehat{l}^{1}_{-i},\dots,\widehat{l}^{T}_{-i}$ is $\eps$-differentially private as a function of $\tau_{i}$.  Since this analysis holds for every player $i$, and shows that the output to all of the remaining players is $(\eps,\delta)$-differentially private as a function of $\tau_{i}$, the entire mechanism (with all the players input) is $(\eps,\delta)$-jointly differentially private.

\begin{comment}
Our goal is to show that each element of the output, $\losses_{i', t}$ can be viewed as a $\lambda$-sensitive query on $u_{i}$.  Since $\noisylosses_{i', t}$ is $\losses_{i', t}$ plus Laplacian noise, it will satisfy differential privacy (for some suitable parameters).  Notice that $\losses_{i',t}$ depends on the utility function $u_{i}$ in two ways.  The first is explicitly, through the action of player $i$.  This can vary arbitrarily with $u_{i}$, but the loss is only $\lambda$-sensitive to this action.  The second is indirect, in that player $i$'s utility function affects the other players' losses, which will in turn affect the query we make.  However, once we fix the noisy losses for the first $t-1$ rounds, we can compute player $i$'s actions in each round, and then only have to worry about the first effect.  Thus, we can view the output of our mechanism as $T$ rounds of (possibly adaptively-chosen) low-sensitivity queries on $u_{i}$, and apply standard composition arguments in that setting.
\end{comment}

\subsubsection{Noisy No-Regret Algorithms Compute Approximate Equilibria}
Therefore we have shown that this ``wrapper" algorithm is jointly differentially private in the sense of Definition \ref{def:privacy}. We now state that using this algorithm with $\swapalg$ will result in an approximate correlated equilibrium.% (Corollary \ref{thm:accCElaplace}).

%\begin{comment}
%\begin{theorem}[Computing CCE]\label{thm:accCCElaplace}
%Let $\nralg = \fixedalg$.  Fix the environment, i.e. the number of players $n$, the number of actions $k$, the largeness of the game $\gamma$, the degree of privacy desired, $(\epsilon, \delta)$, and the failure probability $\beta$.
%One can then select the number of rounds the algorithm must run, $T$, satisfying:
%\begin{equation} \label{eqn:paramassCCE}
%\gamma \eps^{-1} \sqrt{8nkT\log(1/\delta)} \leq \frac{1}{6\log(4nkT/\beta)},
%\end{equation}
%such that with probability at least $1-\beta$, the algorithm $\textsc{NRLaplace}^{\fixedalg}$, returns an $\eta$-approximate CCE for:%
%\footnote{Here $\tilde{O}$ hides (lower order) $\poly( \log n, \log k, \log T, \log(1/\gamma), \log(1/\eps), \log\log(1/\beta), \log\log(1/\delta))$ factors.}
%\begin{equation}\label{eqn:approxCCE}
%\eta = \tilde{O}\left(\frac{\gamma  \sqrt{nk \log(1/\delta)} \log(1/\beta) }{\epsilon}\right).
%\end{equation}
%\end{theorem}
%\end{comment}
%
%
\begin{theorem}[Computing CE]\label{thm:accCElaplace}
Let $\nralg = \swapalg$.  Fix the environment, i.e. the number of players $n$, the number of actions $k$, the largeness parameter of the game $\lambda$, and the degree of privacy desired, $(\epsilon, \delta)$.
One can then select the number of rounds the algorithm $\nrl^{\swapalg}$ must run $T$ and $\beta$  satisfying:
\begin{equation} \label{eqn:paramassCE}
\lambda \eps^{-1} \sqrt{8nkT\log(1/\delta)} \leq \frac{1}{6\log(4nkT/\beta)},
\end{equation}
such that probability at least $1-\beta$, the algorithm returns an $\eta$-approximate correlated equilibrium for:%
%\footnote{Here, $\tilde{O}$ hides lower order $\poly( \log n, \log K, \log T, \log(1/\Delta), \log(1/\eps), \log\log(1/\beta), \log\log(1/\delta))$ factors.}
\begin{equation*}
\eta = \tilde O\left(\frac{ \lambda k \sqrt{n \log(1/\delta) \log(1/\beta)} }{\eps}  \right)
\end{equation*}
\end{theorem}

It is already well known that no-regret alogrithms converge `quickly' to approximate equilibria-- recall Theorems \ref{thm:noregretalgsexist} and \ref{thm:regtoeq}. In the previous section, we showed that adding noise still leads to low regret (and therefore to approximate equilibrium). The tradeoff therefore is this: to get a more `exact' equilibrium, the algorithm has to be run for more rounds but this will result in a less private outcome by the arguments in Theorem \ref{thm:privatecnrl}. The current theorem makes precise the tradeoff between the two.

\begin{theorem}
For every $\lambda$- large game $\G$ with nonnegative utilities bounded by $1$ and action sets $|A_i|\leq k$, there is a mediator $\cM$ such that ``nice behavior" forms an $\eta'$-approximate ex-post Nash equilibrium of the strongly mediated game $\cG_\M^S$ with probability $1-o(1)$ for
$$
\eta' = \tilde O\left( n^{1/4} \sqrt{k \lambda \log(1/\lambda)} \right)
$$
\end{theorem}
\begin{proof}
We choose $\eps = \Theta\left(\sqrt{k \lambda} \left(n \log(1/\beta) \log(1/\delta) \right)^{1/4} \right)$ and $\delta,\beta = \lambda/n$.  With this choice of $\eps,\delta,\beta$ we then choose $T$ to satisfy \eqref{eqn:paramassCE} in the algorithm $\nrl^{\swapalg}$.  From the previous theorem, Theorem \ref{thm:accCElaplace}, we get that $\nrl^{\swapalg}$ computes an $\eta$-approximate correlated equilibrium with probability $1-\beta$, where
$$
\eta = \tilde O\left(  n^{1/4} \sqrt{ k\lambda \log(1/\lambda)} \right).
$$
Note that plugging in these parameter values into Theorem \ref{thm:truth} will get the theorem statement, where $U = 1$.
\end{proof}

This is a \emph{very positive}
result---in several large games of interest, e.g. anonymous matching games, $\lambda = O(n^{-1})$. Therefore, for games of this sort $\eta = \tilde O\left(\frac{ \sqrt{k \log(n)}}{n^{1/4}} \right)$.  Therefore, if $k$ is fixed, but $n$ is large, a relatively exact equilibrium of the underlying game can be implemented, while still being jointly differentially private to the desired degree.  Note that our results remain nontrivial (treating $k$ as fixed) for $\lambda = o\left( \frac{1}{\sqrt{n}\log(n) } \right)$.

\section{Discussion} \label{sec:discussion}
In this work, we have introduced a new variant of differential privacy (joint differential privacy), and have shown how it can be used as a tool to construct mediators which can implement an equilibrium of full information games, even when the game is being played in a setting of incomplete information. As a bonus, our privacy solution concept maintains the property that no coalition of players can learn (much) about any player's type, even in the worst case over whatever prior information they may have had, and thus players have almost no incentive not to participate even if they view their type as sensitive information. Our mediators have little power
in most respects (they cannot enforce actions, they cannot make payments or charge fees, they cannot compel participation). However, for our results which hold for all large games, we require a strong mediator -- i.e. we make the assumption that player types are verifiable in the event that they choose to opt into the mediator. This assumption is reasonable in many settings: for example, in financial markets, there may be legal penalties for a firm misrepresenting relevant facts about itself, or e.g. in an organ matching setting, the mediator may be able to verify blood types of individuals opting in. For our results which hold in the special case of large congestion games, we are able to relax even this assumption, and give mediators which incentivize good behavior even if agents have the ability to misrepresent their types.

At this stage, it is important to ask what more a mediator might be able to achieve at this level of generality. We conclude this paper by considering a few obvious directions.

\subsection{Exact Versus Approximate Equilibria}
A first direction concerns the fact that our mediators aim for approximate equilibria of the resulting game, rather than exact equilibria. The reason for this is that even in large games, it is not possible to coordinate exact equilibria while incentivizing ``good behavior''. For a simple example, consider a game with $2$ players, each of whom must simultaneously choose between two actions, evocatively called ``Mountain'' or ``Beach.'' Player $1$ may be of two types $M$ or $B$. A $M$ (respectively, $B$) type player gets a utility of $1$ (respectively, $0$) if he goes to the mountain, and $0$ (respectively , $1$) he goes to the beach, less a disutility of $\epsilon$ if player $2$ takes the same action. Player $2$ has no private type, and gets $0$ utility from either action, plus a bonus of $\epsilon$ if he matches player $1$'s action. It is easy to verify that the unique exact Nash equilibrium (and correlated equilibrium) of the two full-information games is for player $1$ to take his preferred action, and for player $2$ to match him. It is equally easy to verify that a weak mediator cannot achieve exact Nash equilibrium in this setting, despite the game satisfying our assumptions---player $1$ has an incentive to misreport his type, and send player $2$ to the other location.  A simple variant of this game (with $2n$ player $1$'s and a single player $2$) verifies the same point for a strong mediator at a profile where there are an equal number of both types.

%%%%%%%%%%%%%%%%%%%%%%%%%%%%%%%%%%%%%%%%%%%%%%%%%%%%%%%%%%%%%%%%%%%%%%%%%%%%%%
\subsection{Improving Bounds} \label{sec:lb}
%%%%%%%%%%%%%%%%%%%%%%%%%%%%%%%%%%%%%%%%%%%%%%%%%%%%%%%%%%%%%%%%%%%%%%%%%%%%%%
Even if exact equilibria are beyond reach, one may wonder whether more exact equilibria (or alternately, ones which have a lower approximation error than ours) are achievable. We provide suggestive evidence that this is not possible.

Specifically, we show that there is no algorithm that \emph{privately} computes an $\eta$-approximate equilibrium of an arbitrary $n$-player $2$-action game, for $\eta \ll 1/\sqrt{n \log n}$. In other words, there cannot exist an algorithm that privately computes a ``significantly'' more exact equilibrium without assuming further structure on the game. Thus, the construction of a more exact mediator that provides incentives will therefore have to use different techniques entirely than those in our paper.

%In the case where $\lambda = O(1/n)$ and $k = O(1)$, our algorithm from the previous section computes a differentially private, $\eta$-approximate equilibrium for $\eta \sim 1/\sqrt{n}$ (ignoring all other parameters).  It is natural to ask whether or not we can achieve significantly smaller values of $\eta$ using some other algorithm.

%In this section we prove a lower bound showing that this is not the case.

Our proof is by a reduction to the problem of differentially private \emph{subset-sum query release}, for which strong information theoretic lower bounds are known~\cite{DN03, DY08}.  The problem is as follows:  Consider a database $\mathbf{d} = (d_1, \dots, d_n) \in \bits^n$.  A subset-sum query $q \subseteq [n]$ is defined by a subset of the $n$ database entries and asks ``What fraction of the entries in $\mathbf{d}$ are contained in $q$ and are set to $1$?''  Formally, we define the query $q$ as $q(\mathbf{d}) = \frac{1}{n} \sum_{i \in q} d_i$.  Given a set of subset-sum queries $\cQ = \set{q_1, \dots, q_{|\cQ|}}$, we say that an algorithm $\cM(\mathbf{d})$ \emph{releases $\cQ$ to accuracy $\eta$} if $\cM(\mathbf{d}) = (a_1, \dots, a_{|\cQ|})$ such that $|a_j - q_j(\mathbf{d})| \leq \eta$ for every $j \in [|\cQ|]$.

%\cite{DN03} showed that any differentially private algorithm that releases sufficiently many subset-sum queries must add a significant amount of noise.  A quantitative improvement of their result is given by \cite{DY08}.  They constructed a family $\cQ_{\mathsf{DY}}$ of size $m = O(n)$ such that there is no differentially private algorithm that releases $\cQ_{\mathsf{DY}}$ to accuracy $o(1/\sqrt{n})$.

We show that an algorithm for computing approximate equilibrium in arbitrary games could also be used to release arbitrary sets of subset-sum queries accurately. The following theorem shows that a differentially private mechanism to compute an approximate equilibrium implies a differentially private algorithm to compute subset-sums.

\begin{theorem} \label{thm:lb}
For any $\eta > 0$, if there is an $(\eps, \delta)$-jointly differentially private mechanism $\cM$ that computes an $\eta$-approximate coarse correlated equilibria in $(n+|\cQ| \log n)$-player, $2$-action, $1/n$-sensitive games, then there is an $(\eps, \delta)$-differentially private mechanism $\cM'$ that releases $36\eta$-approximate answers to any $|\cQ|$ subset-sum queries on a database of size $n$.
\end{theorem}

Applying the results of \citeasnoun{DY08}, a lower bound on equilibrium computation follows easily.
\begin{corollary} \label{cor:lb}
Any $(\eps = O(1), \delta = o(1))$-differentially private mechanism $\cM$ that computes an $\eta$-approximate coarse correlated equilibria in $n$-player $2$-action games with $O(1/n)$-sensitive utility functions must satisfy $\eta = \Omega(1/\sqrt{n\log n})$.
\end{corollary}

\subsection{Welfare}
The fact that a mediator achieves \emph{some} equilibrium of the underlying full-information game may be viewed as unsatisfactory. Can a mediator do more, e.g. achieve high (utilitarian) welfare equilibria of the underlying game? The answer is mixed. For general games, the answer is negative. In a recent paper, \citeasnoun{barman} show that it is computationally hard to compute an approximate correlated equilibrium of a game while offering any non-trivial welfare guarantee. In other words, any tractable algorithm which computes a correlated equilibrium of a general game cannot offer non-trival welfare guarantees (unless P = NP).

Within specific classes of games, however, our mediators achieve high welfare. For example, \citeasnoun{Roughgarden09} shows that for the class of \emph{smooth games}, any approximate correlated equilibrium achieves welfare that is a significant fraction of the optimal welfare of any (possibly non-equilibrium) action profile. For such games, of course, our mediators also achieve this welfare guarantee.

\subsection{Strong vs. Weak Mediators }
Next, one may wonder whether there need be any distinction between ``strong'' and ``weak'' mediators, which differ in their ability to verify reported player types. Mechanically, we derive weak mediators by giving algorithms which privately compute Nash equilibria, and strong mediators by giving algorithms which privately compute correlated equilibria. Is the difference merely a failure in our analysis, i.e., might an algorithm that privately computes a correlated equilibrium actually yield a weak mediator?

The answer is no. The intuitive reason is that because joint differential privacy does not constrain the relationship between an agent's report and the suggested action he is given. If the mediator is merely guaranteed to be computing a \emph{correlated equilibrium}, it can be that by misreporting his type, an agent can obtain a suggestion that is more informative about his opponent's actions, which may allow for profitable ``double deviations'' (misreport type and then not follow suggested action). On the other hand, if the mediator is computing a Nash equilibrium, an agent's reported type contains no information about his opponent's randomness, which avoids this problem and yields a weak mediator.

Consider the following simple example of a two player game with two actions -- Beach ($B$) and Mountain ($M$). Player 1 has a fixed type and his utility function is a constant function, independent of the action that either he or his opponent plays. Player 2, in contrast, may be of one of two types: Social ($S$) or antisocial ($A$). Independently of type, player $2$ gets a payoff of $\frac12$ by playing $B$ over $M$. However, if player $2$ is of the social type, then he obtains additional payoff of $1$ if and only if he chooses the same action that his opponent plays. Conversely, if player $2$ is of the antisocial type, then he obtains additional payoff of $1$ if and only if he chooses the opposite action of his opponent.

Now consider the following mediator: independently of player $2$'s report, the mediator suggests to player $1$ that he randomize uniformly between $B$ and $M$. If player $2$ reports type $S$, then the mediator suggests to player $2$ that he play action $B$. If player $2$ reports type $A$, then the mediator always suggests to player $2$ that he play the opposite action of player $1$. Note several things about this mediator: first, it satisfies perfect joint differential privacy, because player $1$'s suggestion is completely independent of player $2$'s report. Second, it always suggests an (exact) correlated equilibrium of the game induced by the reported types. Finally, if player $2$ is of type $S$, he has substantial incentive to deviate from ``good behavior'': he can misreport his type to be $A$, and then always perform the action that the mediator does \emph{not} suggest to him. This results in an expected utility of $5/4$, compared to an expected utility of $1$ if he behaves truthfully. Making the game large does not change the result: a game with a large, odd number of player $1$'s, and a single player $2$ who wants to match/ mismatch the majority, would result in the same result.

The above example demonstrates that it is \emph{not} sufficient for a weak mediator to privately compute a correlated equilibrium, although this is sufficient to implement a strong mediator. We have shown that strong mediators exist in all large games, and that weak mediators exist in all congestion games (and therefore, implicitly, in all potential games that can be concisely represented as potential games). It remains open whether weak mediators exist in all large games.

\bibliographystyle{agsm}
\bibliography{refs}

\newpage
\appendix

\setstretch{1}
\section{Privacy Preliminaries}
We will state some useful tools for achieving differential privacy.  We first cover the Laplace Mechanism \cite{DMNS06}.  
\begin{definition}[Laplace Random Variable]
A continuous random variable $Z$ has a Laplacian distribution with parameter $b>0$ if its probability density $p(z)$ for $Z = z$ is the following:
$$p(z) = \frac{1}{2b} \exp\left(-|z|/b \right).$$
\end{definition}
%We now state some nice properties of the Laplace distribution without proof
%\begin{proposition}
%For a Laplacian random variable $Z$ with parameter $b>0$ we have the following 
%$$
%\Exp[|Z|] = b \qquad \qquad \PROB\left[ |Z| \geq tb\right] = e^{-t}\quad t >0.
%$$
%\label{prop:lap_noise}
%\end{proposition}
%Before we state the Laplacian Mechanism we must define the sensitivity of a function.  
%\begin{definition}[Sensitivity]
%A real function $f:D^n \to \R^k$ has sensitivity $\Delta f$ if we have 
%$$
%\Delta f = \min\left\{\alpha: || f(d) - f(d') ||_1 \leq \alpha, \forall d,d'\in D^n \text{ s.t. } d,d' \text{ are neighboring} \right\}.
%$$
%\end{definition}

We next define the most any single player can influence the output of a function $f$ that takes a database and outputs some vector of real values.  
\begin{definition}[Sensitivity]
The \emph{sensitivity} of a function $f: D^n \to \R^k$ is defined as 
$$
S(f)= \max_{\substack{i \in [n], d_i\neq d_i' \in D \\ d_{-i} \in D^{n-1} }} \left\{ || f(d_i,d_{-i}) - f(d_i',d_{-i}) ||_1\right\}
$$

\end{definition}
If we want to evaluate $f$ on some database and $f$ has large sensitivity, then to make the output of $f$ look the same after a person changes her data, we would need to add a lot of noise to $f$.  We are now ready to present the Laplacian Mechanism $M_L$, given in Algorithm \ref{alg:laplace}.

\begin{algorithm}
\caption{Laplacian Mechanism}\label{alg:laplace}
\begin{algorithmic}[0]
\INPUT : Database $d\in D^n$, query $f: D^n \to \R^k$, and privacy parameter $\epsilon$.
\Procedure {$M_L$} {$d,f, \epsilon$}
\State $\{Z_i\} \stackrel{i.i.d.}{\sim} \text{Laplace}\left(\frac{S(f)}{\epsilon}\right) \quad i = 1, \cdots, k$.
\State $Z \gets (Z_1, \cdots, Z_k)$
\State $\hat f(d) \gets f(d) + Z$
\\
\Return $\hat f(d) $.  
\EndProcedure
\end{algorithmic}
\end{algorithm}

\begin{theorem}
The Laplacian Mechanism $M_L$ is $\epsilon$-differentially private.
\label{thm:laplaceprivacy}
\end{theorem}

The following concentration inequality for Laplacian random variables will be useful.
\begin{theorem}[\cite{GRU12}]
  \label{thm:conc}
  Suppose $\{Y_i\}_{i = 1}^T$ are i.i.d.\ $\Lap(b)$ random variables,
  and scalars $q_i \in [0,1]$. Define $Y :=  \frac{1}{T} \sum_i q_iY_i$. Then for any $\alpha \leq b$,
 \begin{gather*}
  \Pr[ Y \geq \acc ] \leq
      \exp\left( - \frac{\acc^2T}{6b^2} %%% CHANGED THIS 8 FROM A 6
      \right).
  \end{gather*}
\end{theorem}

%To prove Proposition \ref{prop:lap} we can simply look at the ratio of the probability densities for neighboring databases.  We now ask how close is the output of $M_L$ to the actual answer to the query $f(d)$. 
%\begin{proposition}
%Let $f: D^n \to \R^k$ and $\hat f(d) = M_L(d,f, \epsilon)$ be the output of the Laplacian Mechanism.  Then
%$$
%\PROB\left[||f(d) - \hat f(d) ||_\infty \geq\log\left(\frac{k}{\delta} \right) \cdot \left(\frac{\Delta f}{\epsilon} \right) \right] \leq \delta \qquad \delta \in (0,1]
%$$
%\label{prop:util_lap}
%\end{proposition}
%We can prove Proposition \ref{prop:util_lap} using a union bound (for the max $|Z_i|$) and Proposition \ref{prop:lap_noise}.  

 An important result we will use is that differentially private mechanisms `compose' nicely.
 \vfill \pagebreak
%, i.e. that putting together multiple differentially private mechanisms even in an adaptive (rather than fixed) way, still results in a differentially private mechanism.

%\begin{definition}[Adaptive Composition \cite{DRV10}] \label{def:advcomp}
%Let $\bu \in \cU$ be a tuple and $\cA \from \cU \to \cR^{T}$ be any algorithm.  We say $\cA$ is a \emph{$T$-fold adaptive composition of $(\eps, %\delta)$-differentially private mechanisms} if there exists another algorithm $\cB$ such that $\cA(\bu)$ can be written as follows--- for each $t = 1,\dots,T$:
%\begin{enumerate}
%\item $\cB(\cM_1, r_1, \dots, \cM_{t-1}, r_{t-1}) = \cM_{t}$.
%\item $\cM_t$ is an $(\eps,\delta)$-differentially private mechanism.
%\item $r_t$ is a draw according to $\cM_t(\bu)$, $r_t \in \cR$.
%\item The output $\cA(\bu) = (r_1, r_2, \ldots r_{T})$.
%\end{enumerate}
%\end{definition}

%Note here that the choice of the $t^{\text{th}}$ mechanism is not fixed a priori and can depend on the output up to $t-1$. Hence the `adaptive' nomenclature.

\begin{theorem}[Adaptive Composition~\cite{DRV10}] \label{thm:advcomp}
Let $\cA\from \cU \to \cR^{T}$ be a $T$-fold adaptive composition\footnote{See \cite{DRV10} for further discussion} of $(\eps, \delta)$-differentially private mechanisms.  Then $\cA$ satisfies $(\eps', T\delta + \delta')$-differential privacy for
\begin{equation*}
\eps' = \eps \sqrt{2T \ln(1/\delta')} + T\eps(e^{\eps}-1).
\end{equation*}
In particular, for any $\eps \leq 1$, if $\cA$ is a $T$-fold adaptive composition of $\left(\eps/\sqrt{8T\ln(1/\delta)}, 0\right)$-differentially privacy mechanisms, then $\cA$ satisfies $(\eps, \delta)$-differential privacy.
\end{theorem}

\section{Binary Mechanism}

\begin{algorithm}
\caption{Binary Mechanism \cite{Binary}}\label{BinMech}
\begin{algorithmic}[0]
\INPUT : $\omega^{\leq T} \in \{-1,0,1 \}^T$, privacy parameter $\epsilon> 0$
\OUTPUT : Count with Laplacian noise $\hat{y}^T$
\Procedure {$\bm$} {$\omega^{\leq T}, \epsilon'$}
\State {\bf initialize:} $\rr$ is a table of zeros of size $\lfloor \log(T) \rfloor+ 1$ by $T$.
\For{$t = 1, \cdots, T$}
\State Write $t = \sum_{j = 0}^{\lfloor \log(t)\rfloor}b_j(t) 2^j$
\State $ i \gets \min\{j:b_j(t) \neq 0 \} $
\For{$j= 0, \cdots,i $}
	\State $\rr(j,t/2^j) \gets \sum_{\ell = t/2^j}^t \omega^\ell $
	\State $\hat\rr(j,t/2^j) \gets \rr(j,t/2^j) +$ Lap$(1/\epsilon)$
\EndFor
$\hat{y}^t \gets \sum_{j = 0}^{\lfloor \log(t) \rfloor} b_j(t) \hat\rr(j,t/2^j)$
\EndFor
 \State \textbf{return} $ \hat y^{\leq T}$
\EndProcedure
\end{algorithmic}
\end{algorithm}

Note that we have modified the original presentation of the binary mechanism from \cite{Binary} but it has the exact same behavior; we are not concerned with space so we keep a table of entries rather than a vector of entries that continually get overwritten as zero to compute the counts $\hat y^t$ for each $t \in [T]$ in $\bm$.  As presented here, $\bm(\omega^{\leq T},\epsilon)$ would be $(\log(T)+1) \epsilon$- differentially private.  The output of $\bm$ is a vector of counts $\hat y^{\leq T}$, whereas $\abm$ produces a single count for a particular $t \in [T]$.    If we start with $\hat\rr^1$ as a table of zeros, and we define for $t = 2, \cdots, T$
$$
\hat\rr^{t} = \pst(\omega^{\leq t},\hat\rr^{t-1},\epsilon) \qquad \hat y^{t} = \abm(\hat\rr^{t},t)
$$
then given $\omega^{\leq T}$ the vector $\hat y^{\leq T}$ has the same distribution over outcomes as $\bm(\omega^{\leq T},\epsilon)$.  

\section{Omitted Proofs from Weak Mediators Section}
\subsection{Designing a Weak Mediator}
\begin{proof}[Proof of Lemma \ref{lem:Tnoise}]
At round $t\in [T]$, player $i$ is making an $\alpha$-noisy best response to get the updated action profile $\ba^{t+1}$, which means that her actual cost may only be decreasing by as little as $\alpha - 2\mathbf{\Delta}$.
\begin{align*}
c(\tau_i, \ba^{t+1})  & \geq  \hat c(\tau_i,\ba^{t+1}) - \mathbf{\Delta} \geq \hat c(\tau_i,\ba^t) +\alpha - \mathbf{\Delta}  \\
& \geq  c(\tau_i,\ba^t)+\alpha - 2\mathbf{\Delta}
\end{align*}
We use the potential function $\Phi(\mathbf{y(\ba)}) = \sum_{e\in E} \sum_{i = 1}^{y_e(\ba)} \ell_e(i)$ to then bound the number of possible $\alpha$-noisy best responses made by all the players.  Note that the difference in potential when a player changes an action is the same as the cost difference that same player that moved experiences, i.e. $\Phi(\mathbf{y(\ba^{t+1})}) - \Phi(\mathbf{y(\ba^{t})}) = c(\tau_i, \ba^{t+1}) - c(\tau_i,\ba^t)$.
\begin{equation*}
T \leq \frac{mn}{a-2\mathbf{\Delta}}   < 2 \frac{mn}{\alpha}. \qedhere
\end{equation*}
\end{proof}

\begin{proof}[Proof of Claim \ref{claim:one}]
When a player makes an $\alpha$- best response at time $t$ with respect to the noisy counts on each facility, then we can bound the amount he could improve by
\begin{align*}
Improve(\tau_i, \ba^t) & =  c(\tau_i,\ba^t) -\min_{a_i'\in R_i}\{ c(\tau_i, (a_i', a_{-i}^t)) \} = c(\tau_i,\ba^t) -\{ c(\tau_i, (a_i^*, a_{-i}^t)), \\
& \leq \hat c(\tau_i,\ba^t)+ \mathbf{\Delta} -( \hat c(\tau_i, (a_i^*,a_{-i}^t)) - \mathbf{\Delta}), \\
& \leq 0 + 2\mathbf{\Delta}. \qedhere
\end{align*}
\end{proof}

\begin{proof}[Proof of Lemma \ref{lem:noisygap}]
Let $t$ and $t+t'$ be two times that player $i$ of type $\tau_i$ made $\alpha$-noisy best responses.  We have
$$
\alpha \leq Improve(\tau_i, \ba^{t+t'}) \leq 2\mathbf{\Delta} + (t') m\sigma_\ell.
$$
We set $t' = \frac{\alpha -2 \mathbf{\Delta} }{m\sigma_\ell}$.  Further, we know from \eqref{eq:Tnoise} that under the low error assumption, there can be as many as $\frac{2mn}{\alpha}$ total $\alpha$-noisy best responses.  With our bound on the noisy gap we obtain a bound on the number of $\alpha$-noisy best response moves $p$ a single player can make is
$$
p \leq \frac{\frac{2mn}{\alpha}}{\Gamma} \leq \frac{2m^2n\sigma_\ell}{\alpha(\alpha-2\mathbf{\Delta})},
$$
and for $\alpha > 4 \mathbf{\Delta}$ we obtain the stated result for $p$.
\end{proof}
\subsection{Analysis of $\pbr$}
We first state a few general lemmas which we use to prove $\pbr(\tau)$ is jointly-differentially private.
\begin{lemma}[Post-Processing \cite{DMNS06}]
Given a mechanism $\M: \mathcal{T}^n \to \mathcal{O}$ and some function $\phi: \mathcal{O} \to \mathcal{O}'$ that is independent of the players' types $\tau \in \mathcal{T}^n$, if $\M(\tau)$ is $(\epsilon, \delta)$- differentially private then $\phi(\M(\tau))$ is $(\epsilon, \delta)$-differentially private.
\label{lem:comp}
\end{lemma}

\begin{lemma}[Billboard Lemma]
Let $\M: \mathcal{T}^n \to \mathcal{O}^n$ be an $(\epsilon, \delta)$ differentially private mechanism and consider any function $\theta: \mathcal{T} \times \mathcal{O}^n \to \mathcal{A}^n$. Define the mechanism $\M': \mathcal{T}^n \to \mathcal{A}^n$ as follows: on input $\tau$, $\M'$ computes $o = \M(\tau)$, and then $\M'(\tau)$ outputs to each $i$:
$$\M'(\tau)_i = \theta(\tau_i,o).$$
$\M'$ is  then $(\epsilon, \delta)$-jointly differentially private.
\label{lem:billboard}
\end{lemma}
\begin{proof}
We fix a player $i \in [n]$ and let $B_{-i} \subseteq \mathcal{A}^{n-1}$.
\begin{align*}
\PROB_{a_{-i}\sim M'(\tau) }[a_{-i} \in B_{-i}]= \PROB_{o \sim M(\tau)}[(\theta(\tau_j,o) )_{j\neq i}\in B_{-i}].
\end{align*}
We define $\phi_j(o) = \theta(\tau_j,o)$ and then
$$
\phi(o) = (\phi_j(o))_{j\neq i}.
$$
Note that $\phi$ is independent of $\tau_i$.  We can now invoke Lemma \ref{lem:comp}.
\begin{align*}
\PROB_{o \sim M(\tau)}[(\theta(\tau_j,o) )_{j\neq i}\in B_{-i}] &= \PROB_{o \in M(\tau)}[\phi(o) \in B_{-i}],  \\
& \leq \exp(\epsilon)  \PROB_{o \in M(\tau_i',\tau_{-i})}[\phi(o) \in B_{-i}] + \delta, \\
& = \exp(\epsilon) \PROB_{a_{-i}\sim M'(\tau_i',\tau_{-i})_{-i} }[a_{-i} \in B_{-i}]+ \delta. \qedhere
\end{align*}
\end{proof}
We are now ready to prove Theorem \ref{lem:dp},that $\pbr(\tau)$ is indeed jointly differentially private.
\begin{proof}[Proof of Theorem \ref{lem:dp}]
To achieve joint differential privacy, we want to apply Lemma \ref{lem:billboard}.  We first argue that if we were to output the entire transcript of partial counts on every facility after every player has moved, then this is differentially private in the reported types of the players.  We will denote the mechanism $M$ as 
$$M: \mathcal{T}^n \to \R^{\ifaddone \left( \lfloor\log(nT)\rfloor+1\right)\else \lfloor\log(nT)\rfloor \fi \cdot nT \cdot m},  \qquad M\left( \tau \right) = (\hat q_e^{\leq nT})_{e \in E},
$$
where the output of $M$ is the vector form of the final table $\hat \rr^{nT} = (\hat \rr^{nT}_e)_{e \in E}$ computed in $\pbr$ that uses $\pst$ and $\abm$.  Each $\hat q_e^t$ is the vector of entries that were updated in the table $\hat \rr_e^t$ from $\hat \rr_e^{t-1}$.  For ease of notation, we will write values without subscripts to be the vector of values for every subscript, e.g. $\hat q ^{\leq t} = (\hat q_e^{\leq t})_{e \in E}$. 

There are three random vectors $W^t$, $Q^t$ and $\hat Q^t$ that we need to account for at each iteration $t$ of $\pbr$.  Initially, there is a vector stream $\omega^{\leq n} \in \{-1,0,1 \}^{m \times n}$ which is deterministically chosen as a function of what each of the $n$ players report.  We then set $W^t = \omega^t$ for $t \leq n$.  

With each stream $\omega^{\leq t}$ we then use the routine $\pst$, which calculates the table of \emph{exact} partial counts $\rr^t = (\rr^t_e)_{e \in E}$, where $\rr^t_e \in \R^{\left(\lfloor \log(nT) \rfloor+1\right) \times nT}$ and then outputs the table of \emph{approximate} partial counts $\hat \rr^{t}$.  We will use the function $g^t_e: \{ -1,0,1\}^{t} \to \N^{i(t)}$ to denote the vector of \emph{exact} partial counts that get updated at time $t$ in $\pst$ where 
$$
i(t) = \min\left\{j: b_j(t) \neq 0 \quad \text{with} \quad t = \sum_{\ell=0}^{\lfloor\log(t) \rfloor} b_\ell(t)2^\ell\right\}.
$$ 
That is, 
$$
g_e^t\left(\omega_e^{\leq t} \right)= \left(\rr^t_e(j,t/2^{j}) \right)_{j = 0}^{i(t)} = \left( \sum_{\ell = t/2^{j}}^{t} \omega_e^{\ell} \right)_{j = 0}^{i(t)}, \qquad \text{ and } \qquad g^t(\omega) = (g^t_e(\omega_e) )_{e \in E}.
$$
We then define the variable
\begin{equation}
Q^t = g^t(W^{\leq t}),
\label{eq:Q}
\end{equation}  
which up to round $n$ is deterministic.  We introduce the noise vector $Z^t$ of size $m \cdot i(t)$ with each entry being an i.i.d. Laplace random variable with parameter $1/\epsilon'$.
\begin{equation}
\hat Q^t = Q^t + Z^t.
\label{eq:hatQ}
\end{equation}
We note that there is a one to one mapping $\phi$ that takes any output of $\pst(\omega_e^{\leq t},\hat \rr_e^{t-1}) = \hat \rr_e^t$ that does not include any of the zero entries that have not been updated from $\rr^1_e$ and $\rr_e^t$, to a realization of $\hat Q^{\leq t} = \hat q^{\leq t}$; the quantity $\hat q_e^{\leq t}$ is just the table $\hat \rr_e^t$ represented as a vector so that $\phi(\hat\rr_e^t) = \hat q_e^{\leq t}$.  It will be more convenient in this proof to focus on the vector $\hat q^{\leq t}$, rather than the table $\hat\rr^t$.  

We are left with defining $W^t$ for $t>n$.  In $\pbr$, each bit $\omega_e^t$ for $t > n$ depends on which player is moving at time $t$ and the table of \emph{approximate} partial counts, $\hat \rr^{t-1}$, from which we can count the number of people $\hat y_e^{t-1}$ on facility $e$ at that time.  We denote the player that moves at time $t$ when they reported type $\tau$ as $\tau(t)$.  The bit $\omega_e^t$ is then determined by player $\tau(t)$'s best response move given the table $\hat \rr^{t-1}$ (or equivalently $\hat q^{\leq t-1}$).  We denote this dependence as a function $f^t(\hat q^{\leq t-1}; \tau(t) )$.  We then define the random variable $W^t$ as 
\begin{equation}
W^t = f^t\left(\hat Q^{\leq t-1} ; \tau(t)\right) \qquad \text{ for } t>n.
\label{eq:W}
\end{equation}
We then have that $M(\tau)$ and $\hat Q^{\leq nT}$ have the same distribution.  

When proving differential privacy, we set the outcome and see how much the probability distributions differ when we change a single entry in our database.  If we change the reported types $\tau$ to $(\tau_i',\tau_{-i})$, then the random variable $W^t$ changes to $(W')^t$, $Q^t$ changes to $(Q')^t$, and $\hat Q^t$ changes to $(\hat{Q'})^t$ but we fix the outcome realization $\hat q^{\leq nT} $ in both cases.  

Regardless of whether player $i$ reports $\tau_i$ or $\tau_i'$, we know that he can move at most $p$ different times.  We let $\{t_1, \cdots, t_p \}$ be the times when $i$ moves when reporting $\tau_i$ given realization $\hat q^{\leq nT}$ and $\{t_1', \cdots, t_p'  \}$ be the times when $i$ moves having reported $\tau_i'$ instead.  We then relabel the times and reorder them so that $t_1 < t_2< \cdots, t_{2p}$ are the times when $i$ moves having reported $\tau_i$ or $\tau_i'$. We denote $T^p = \{t_1, \cdots, t_{2p} \}$.  We consider bounding the following quantity
\begin{align*}
\Prob{}{\hat Q^{\leq nT} = \hat q^{\leq nT} }  = \prod_{t = 1}^{nT} \Prob{}{\hat Q^{t} = \hat q^{t} \left| \hat q^{< t} \right.}.
 \end{align*}
From \eqref{eq:W} we know $W^t = (W')^t$ when we condition on $\hat q^{< t}$ and $t \notin T^p$.  We then consider rounds $t \in T^p$, while conditioning on $\hat q^{<t}$. We know that each realization of $W^{t}_e - (W')^{ t}_e \in \{-1,0,1 \}$ for each $e \in E$.  Thus, $W^{\leq nT}$ will differ by one in as many as $2pm$ entries from $(W')^{\leq nT}$ when we condition on $\hat q^{< t}$ for each entry $t \leq nT$.  Hence, $Q^{\leq nT}$ has total $\ell_1$ distance $2pm \left(\log(nT) + 1 \right)\leq 3pm\log(nT)$ from $(Q')^{\leq nT}$ when we again condition on $\hat q^{<t}$ for each round.  We then add independent Laplacian noise with parameter $1/\epsilon'$ to each component to get
\begin{align*}
\Prob{}{\hat Q^{\leq nT} = \hat q^{\leq nT} } & = \prod_{t = 1}^{nT} \Prob{}{Z^t =  \hat q^{t} - Q^{t} \left| \hat q^{< t} \right.}, \\
&\leq \exp(3\epsilon 'pm\log(nT))\cdot \prod_{t = 1}^{nT} \Prob{}{Z^t =  \hat q^{t} - (Q')^{t} \left| \hat q^{< t} \right.}, \\
&= \exp(3\epsilon 'pm\log(nT))\cdot \Prob{}{(\hat Q')^{\leq nT} = \hat q^{\leq nT} }= \exp(\epsilon) \Prob{}{(\hat Q')^{\leq nT} = \hat q^{\leq nT} }.
 \end{align*}
We have thus shown that $M$ is $ \epsilon$-differentially private.

From the partial sum tables that $M$ outputs and the extra bit of information to each player $j$ being which initial action $\pbr$ had $j$ taking, each player can determine what actions the algorithm had her moving to at each round that she moved and hence the last action she was taking when the algorithm terminated.  We then use Lemma \ref{lem:billboard} to conclude that we have an $\epsilon$-joint differentially private mechanism.  Lastly, we get $(\epsilon, \beta)$ joint differential privacy (instead of $(\epsilon, 0)$ joint differential privacy) because there is a small probability $\beta$ that $\pbr$ may not produce an output, which may be disclosive.  
\end{proof}

We now prove the remaining theorems from this section.

\begin{proof}[Proof of Theorem \ref{thm:nash_eq}]
  We first assume that the low error assumption holds. Recall that in this case, no player ever makes more than $p$ action changes, and so the algorithm outputs an action profile and does not fail. We consider the difference in cost to player $i$ when she deviates from her suggested action $a_i$ to a best response action $a^*_i$ given the other players' actions $a_{-i}$.  Recall that $\mathbf{\Delta}$ is the upper bound on the error between the noisy costs for each player and their true cost.  We showed in Lemma \ref{lem:Tnoise} that given $\alpha > 4\mathbf{\Delta}$ the total number of possible $\alpha$-noisy best response moves in a routing game is $T$.  $\pbr (\tau)$ iterates through each of the $n$ players $T$ times (for a total of $nT$ time steps) which is enough to elicit $T$ $\alpha$-noisy best responses.  Under the low error assumption therefore, at completion with action profile $\mathbf{a}$, no player can have an $\alpha$-noisy best response. We have,
\begin{align*}
\eta  & = \max_{i \in [n]}\{Improve(\tau_i, \ba)\} = \max_{i \in [n]}\{c(\tau_i,\ba) - c(\tau_i, (a_i^*,a_{-i}))\} \leq \alpha + 2 \mathbf{\Delta}.
\end{align*}
We now expand $\mathbf{\Delta} $ from \eqref{eq:costdiff} by substituting in the privacy parameter $\epsilon'$ that was given in Algorithm \ref{BR-PRIVATE}.  We have bounds for the number of total approximate best response moves $T$ given in \eqref{eq:Tnoise} and the total number of times $p$ any single player can change actions, found in \eqref{eq:knoise}, which are used in $\pbr(\tau)$:
\begin{align*}
\eta \leq \alpha + 2 \mathbf{\Delta}  = \alpha + O\left(m^4  n (\sigma_\ell)^2 \frac{\log^{3/2}\left( m n/\alpha\right)\sqrt{\log(m/\beta)}}{\alpha^2\epsilon} \right).
\end{align*}
If we set
$$
\alpha = \Theta\left( \left(\frac{m^4n (\sigma_\ell)^2\log^{2}(mn/\beta) }{\epsilon}\right)^{1/3}\right),
$$
we can balance the two terms in the bound of $\eta$ to get the desired result.  We still need to check the condition $\alpha > 4 \mathbf{\Delta} $.  This requires
$$
\alpha^3 > 48 \sqrt{8} \frac{  m^4 n(\sigma_\ell)^2 \log^{2}(2 mn/(\alpha \beta))}{\epsilon}.
$$
This is precisely the ordering for $\alpha$ we found to obtain the best bound for $\eta$.  To complete the proof, we note that the low error assumption holds with probability at least $1-\beta$.
\end{proof}

\begin{proof}[Proof of Theorem \ref{thm:weak_main}]
We have satisfied the hypotheses of Theorem \ref{lem:main} (with the slight difference that we are dealing with costs instead of utilities) because we have an algorithm $\M(\tau)  = \pbr(\tau)$ that is $(\epsilon, \beta)$ joint differentially private and with probability $1-\beta$ produces an action profile that is an $\eta$ approximate pure strategy Nash Equilibrium of the routing game defined by the type vector $\tau$.  Thus for any $\tau \in \mathcal{T}^n$, the ``good" behavior profile $(\tau,f)$ is an $\eta'$-approximate Nash Equilibrium of the complete information game $\G_\M$, where
\begin{align*}
\eta'  \leq & \eta + m\epsilon + m\beta + m \beta = O\left( \left(\frac{m^4n (\sigma_\ell)^2\log^{2}(mn/\beta)}{\epsilon}\right)^{1/3}\right) + m\epsilon+ 2m\beta .
\end{align*}
We can now optimize the bound for $\eta'$ by adjusting $\epsilon$.  We set
\begin{align}
\epsilon = \Theta \left( \left(mn (\sigma_\ell)^2 \log^{2}(mn/\beta) \right)^{1/4}\right).
\label{epsilon}
\end{align}

We then need to scale $\beta$ in an optimal way recalling that the bound in the difference in the exact and approximate counts in Lemma \ref{lem:chanerror} requires $\beta>2m/(nT)$ (note that the stream is length $nT$ instead of $T$). %which is satisfied when $\alpha = O\left(\frac{(n \Delta\ell)^4 \log^3(n)}{m} \right)$ which holds for our value of $\alpha$ in \eqref{eq:alpha}.  
We set $\beta = \sqrt{\sigma_\ell} / n$ for $n$ sufficiently large.
This gives our bound for $\eta'$.
\end{proof}

\section{Proofs of Noise Tolerance of No Regret Algorithms (Section \ref{sec:no-regret})}

\begin{proof}[Proof of Lemma \ref{lem:noisyregretsbounded}]
Let $(\state_0, \dots, \state_{T})$ be any sequence of distributions and $f \from \actionset \to \actionset$ be any function.  Then:
\begin{align}
\regret(\state_0, &\dots, \state_T, \Losses, f) - \regret(\state_0, \dots, \state_T, \noisyLosses, f) \notag\\
={} &( \exploss(\state_0, \dots, \state_T, \Losses) - \exploss(\modstate_0, \dots, \modstate_T, \Losses) ) - ( \exploss(\state_0, \dots, \state_T, \noisyLosses) - \exploss(\modstate_0, \dots, \modstate_T, \noisyLosses) ). \notag \\
={} & ( \exploss(\state_0, \dots, \state_T, \Losses) - \exploss(\state_0, \dots, \state_T, \noisyLosses) ) + ( \exploss( \modstate_0, \dots, \modstate_T, \noisyLosses) - \exploss(\modstate_0, \dots, \modstate_T, \noisyLosses) ) \notag\\
={} & \left( \frac{1}{T} \sum_{t = 1}^{T} \sum_{j = 1}^{k} \tjstate (\tjloss - \noisytjloss) \right) +
\left( \frac{1}{T} \sum_{t=1}^{T} \sum_{j=1}^{k} (\modstate_t)^j (\tjloss - \noisytjloss) \right) \qquad (\text{by definition of } \exploss)\notag \\
={} & \left( \frac{1}{T} \sum_{t = 1}^{T} \sum_{j = 1}^{k} \tjstate \tjnoise \right) +
\left( \frac{1}{T} \sum_{t=1}^{T} \sum_{j=1}^{k} (\modstate_t)^j \tjnoise \right) \,\quad\qquad\qquad\qquad (\text{by definition of }z)\label{eq:noisyregrets0} \\
\leq{}& \zeta \left( \frac{1}{T} \sum_{t=1}^{T} \sum_{j=1}^{k} \tjstate \right) +  \zeta \left( \frac{1}{T} \sum_{t=1}^{T} \sum_{j=1}^{k} (\modstate_t)^j \right) \,\,\,\qquad\qquad\qquad\qquad (\forall j,t \; |z^{j}_{t}| \leq \zeta)\nonumber\\
={}& 2\zeta, \nonumber
\end{align}
where the final equality follows from the fact that $\tstate, \modstate_t$ are probability distributions.
\end{proof}

\begin{proof}[Proof of Corollary \ref{thm:lowregretbounded}]
We will prove only item $1$, the proof for $2$ is analogous. First, by the assumption of the theorem, we will have $\noisyLosses \in [0,1]^{T \times k}$ except with probability at most $\beta$.
Therefore, by Theorem~\ref{thm:noregretalgsexist},
\begin{align*}
\Prob{\Noises}{\regret\left(\overrightarrow{\fixedalg}, \noisyLosses, \fixedmods\right) > \sqrt{\frac{2 \log k}{T}}} \leq \beta
\end{align*}
Further, by Lemma \ref{lem:noisyregretsbounded}, we know that $\noisyLosses \in [0,1]^{T \times k}$ implies
\begin{align*}
\regret\left(\overrightarrow{\fixedalg}(\noisyLosses), \Losses, \mods\right)  \leq \regret\left(\overrightarrow{\fixedalg}, \noisyLosses, \mods\right) + 2\zeta.
\end{align*}
Combining, we have the desired result, i.e.
\begin{align*}
&\Prob{\Noises}{\regret\left(\overrightarrow{\fixedalg}(\noisyLosses), \Losses, \fixedmods\right) > \sqrt{\frac{2 \log k}{T}} + 2\zeta} \leq \beta. \qedhere
\end{align*}
\end{proof}

\begin{proof}[Proof of Lemma \ref{lem:noisylapregretsbounded}]
Let $(\state_0, \dots, \state_{T})$ be a sequence of distributions and $f \from \actionset \to \actionset$ be any function.  Recall by~\eqref{eq:noisyregrets0},
\begin{align}
&\regret(\state_0, \ldots, \state_T, \Losses, f) - \regret(\state_0, \dots, \state_T, \noisyLosses, f)
= \left( \tfrac{1}{T} \sum_{t = 1}^{T} \sum_{j=1}^{k} \tjstate \tjnoise \right) +\left( \tfrac{1}{T} \sum_{t=1}^{T} \sum_{j=1}^{k} (\modstate_t)_j \tjnoise \right).
\label{eq:noisyregrets}
\end{align}
We wish to place a high probability bound on the quantities:
\begin{align*}
&Y_{\state_0, \dots, \state_{T}} =   \frac{1}{T} \sum_{t = 1}^{T} \sum_{j = 1}^{k} \tjstate \tjnoise.
\intertext{Changing the order of summation,}
&Y_{\state_{0}, \dots, \state_{T}}
={} \sum_{a_{1}, \dots, a_{T} \in A} \left( \prod_{t=1}^{T} \state_{t}^{a_{t}} \right) \left(\frac{1}{T} \sum_{t=1}^{T} z_{t}^{a_{t}} \right),
\end{align*}
the equality follows by considering the following two ways of sampling elements $z^j_{t}$.  The first expression represents the expected value of $z^{j}_{t}$ if $t$ is chosen uniformly from $\set{1,2,\dots,T}$ and then $j$ is chosen according to $\pi_{t}$.  The second expression represents the expected value of $z^{j}_{t}$ if $(a_1, \dots, a_{T})$ are chosen independently from the product distribution $\pi_{1} \times \pi_{2} \times \dots \times \pi_{T}$ and then $a_{t}$ is chosen uniformly from $(a_1, \dots, a_{T})$.  These two sampling procedures induce the same distribution, and thus have the same expectation.
Thus we can write:
\begin{align*}
\Prob{\Noises}{Y_{\state_{0}, \dots, \state_{T}} > \alpha}
\leq \max_{a_{1}, \dots, a_{T} \in A}  \Prob{\Noises}{ \frac{1}{T} \sum_{t=1}^{T} z_{t}^{a_{t}} > \alpha}
\leq  \Prob{\Noises}{ \frac{1}{T} \sum_{t=1}^{T} z_{t}^{1} > \alpha}.
\end{align*}
where the second inequality follows from the fact that the variables $z_{t}^{j}$ are identically distributed. Applying Theorem~\ref{thm:conc}, we have that for any $\alpha < \zeta$,
\begin{align} \label{eqn:conc1}
\Prob{\Noises}{Y_{\state_{0}, \dots, \state_{T}} > \alpha} \leq e^{-\alpha^2 T / 6b^2}.
\end{align}
Let $(\state_{0}, \dots, \state_{T}) = \nralg(\noisyLosses)$.  By Equation~\eqref{eq:noisyregrets} we have
\begin{align*}
&\Prob{\Noises}{\regret(\nralg(\noisyLosses), \Losses, f) - \regret(\nralg(\noisyLosses), \noisyLosses, f) > \alpha} \\
\leq{} &\Prob{\Noises}{ \frac{1}{T} \sum_{t = 1}^{T} \sum_{j = 1}^{k} \tjstate \tjnoise > \alpha/2} + \Prob{\Noises}{\frac{1}{T} \sum_{t=1}^{T} \sum_{j=1}^{k} (\modstate_t)_j \tjnoise > \alpha/2}
\leq{} 2e^{-\alpha^2 T/24b^2}
\end{align*}
where the last inequality follows from applying (\ref{eqn:conc1}) to the pair of sequences $(\state_0, \dots, \state_{T})$ and $(\modstate_{0}, \dots, \modstate_{T})$. The Lemma now follows by taking a union bound over $\mods$.
\end{proof}

\begin{proof}[Proof of Corollary \ref{thm:lowregretlaplace}]
First, we demonstrate that $\noisyLosses \in [0,1]^{T \times k}$ except with probability at most $\beta$, which will be necessary to apply the regret bounds of Theorem~\ref{thm:noregretalgsexist}.  Specifically:
\begin{equation} \label{eqn:beta1}
\Prob{\Noises}{\exists \tjnoise \textrm{ s.t. } |\tjnoise| > \frac{1}{3}} \leq Tk \Prob{\Noises}{|\noise^1_{1}| > \frac{1}{3}} \leq 2Tke^{-1/6b} \leq \beta/2,
\end{equation}
where the first inequality follows from the union bound, the second from the definition of Laplacian r.v.'s and the last inequality follows from the assumption that $b \leq 1/6 \log(4Tk/\beta)$.

The theorem now follows by conditoning on the event $\noisyLosses \in [0,1]^{T \times k}$ and combining the regret bounds of Theorem~\ref{thm:noregretalgsexist} with the guarantees of Lemma \ref{lem:noisylapregretsbounded}. For parsimony, we will only demonstate the first inequality, the second is analogous. Recall again by Theorem \ref{thm:noregretalgsexist}, we have that whenever $\noisyLosses \in [0,1]^{T \times k}$:
\begin{align}
&\regret(\overrightarrow{\fixedalg}, \noisyLosses, \fixedmods) \leq \sqrt{\frac{2 \log k}{T}}. \nonumber
\intertext{Further, by Lemma \ref{lem:noisylapregretsbounded}, we know that:}
&\Prob{\Noises}{ \regret\left(\overrightarrow{\fixedalg}(\noisyLosses), \Losses, \fixedmods\right) - \regret\left(\overrightarrow{\fixedalg}, \noisyLosses, \fixedmods\right)  > \alpha} \leq 2|\fixedmods| e^{- \alpha^2T/ 24b^2}\nonumber\\
&\hphantom{\Prob{\Noises}{ \regret\left(\fixedalg(\noisyLosses), \Losses, \fixedmods\right) - \regret\left(\fixedalg(\noisyLosses), \noisyLosses, \fixedmods\right)  > \lambda}} = 2k e^{- \alpha^2 T/ 24b^2}.\nonumber\\
\intertext{Substituting $\alpha = b\sqrt{\frac{24\log(4k/\beta)}{T}}$, we get:}
&\Prob{\Noises}{ \regret\left(\overrightarrow{\fixedalg}(\noisyLosses), \Losses, \fixedmods\right) - \regret\left(\overrightarrow{\fixedalg}, \noisyLosses, \fixedmods\right)  > \alpha} \leq \beta/2. \label{eqn:beta2}
\end{align}
The result follows by combining (\ref{eqn:beta1}) and (\ref{eqn:beta2}).
\end{proof}

\begin{proof}[Proof of Lemma \ref{lem:lossscaled}]
Let $\state_0, \dots, \state_T \in \actiondists$ be any sequence of distributions and let $f\from \actionset \to \actionset$ be any function.  Then
\begin{align*}
\regret(\state_0, \dots, \state_T, \Losses, f)
&= \exploss(\state_0, \dots, \state_T, \Losses) - \exploss(\modstate_0, \dots, \modstate_T, \Losses) \\
&= 3\left(\exploss(\state_0, \dots, \state_T, \scaledLosses) - \exploss(\modstate_0, \dots, \modstate_T, \scaledLosses) \right) \\
&= 3 \left( \regret(\state_0, \dots, \state_T, \scaledLosses, f) \right).
\end{align*}
The second equality follows from the definition of $\exploss$ and from linearity of expectation.
The Lemma now follows by setting $(\state_0, \dots, \state_T) = \vec{\nralg}(\scaledLosses)$, %changed \nralg_T
taking a maximum over $\mod \in \mods$, and plugging in the guarantees of Theorem~\ref{thm:noregretalgsexist}.
\end{proof}

\section{Proofs for Computing Equilibria in Games with Few Actions (Section \ref{sec:noisyDP})}
\begin{proof}[Proof of Theorem~\ref{thm:privatecnrl}]

Fix any player $i$, any pair of types for $i$, $\tau_{i}, \tau'_{i}$, and a tuple of types $\tau_{-i}$ for everyone else. To show differential privacy, we need to analyze the change in the distribution of the joint output for all players other than $i$, $(\state_{-i,1}, \dots, \state_{-i,T})$ when the input is $(\tau_i, \tau_{-i})$ as opposed to $(\tau'_i,\tau_{-i})$.

It will be easier to analyze the privacy of a modified mechanism that outputs $ (\noisylosses_{-i,1}, \dots, \noisylosses_{-i,T})$.  Observe that this output is sufficient to compute $(\state_{-i,1}, \dots, \state_{-i,T})$ just by running $\nralg$.  Thus, if we can show the modified output satisfies differential privacy, then the same must be true for the mechanism as written.

Let us define $\pi_{j,1}$ to be the uniform distribution over the $k$ actions for player $j \in [n]$, which is independent in the input types of all the players.  We then define $l_{j,1}^a$ as the loss for player $j \in [n]$ for action $a \in \actionset$ at round 1 for input type $(\tau_i, \tau_{-i})$ and similarly ${l'}_{j,1}^a $ for input type $(\tau_i'.,\tau_{-i})$.  Note that $l_{j,1}^a = {l'}_{j,1}^a$ for all $j \neq i$ and $a \in \actionset$.  

We then define the random variables $\hat{L}_{j,1}^a = l_{j,1}^a + Z_{j,t}^a$ and $\hat{L'}_{j,1}^a = {l'}_{j,1}^a + Z_{j,1}^a$where $Z_{j,1}^a$ is a Laplace random variable with parameter $b$.  We are then outputting a realization $(\hat{l}_{-i,1}^a )_{a \in [k]}$ that we hold fixed no matter if the input type is $(\tau_i,\tau_{-i})$ or $(\tau_i',\tau_{-i})$.  Note that $\hat{L}_{j,1}^a$ and $\hat{L'}_{j,1}^a$ have the same distribution for $j \neq i$.  

We then compute the distribution over actions for the next iteration $\pi_{j,2} = \cA(\pi_{j,1},\hat{l}_{j,1})$, for all $j\neq i$.  Note that because we are not fixing the output for the $i$'th player, we have $\pi_{i,2} \neq {\pi'}_{i,2}$ and can change arbitrarily.  We then ask how much this can effect the other $j \neq i$ players' losses at round 2
\begin{align*}
 l_{j,2}^a  & = 1- \Ex{\ba_{-j} \sim \state_{-j,2}}{u(\tau_{j},(a, \ba_{-j}))} ={} 1 - \Ex{\ba_{-(i,j)} \sim \state_{-(i,j),t}}{\Ex{a_i\sim \state_{i,t}}{u(\tau_{j},(a,a_i,\ba_{-(i, j)}))}}
\\ & \leq 1 - \Ex{\ba_{(-i,j)} \sim \state_{-(i,j),t}}{\Ex{a_i \sim {\state}_{i,t}}{u(\tau_{j},(a,a_i,\ba_{-(i, j)})) + \lambda}} = {l'}_{j,2}^a  + \lambda,
\end{align*} 
where the inequality comes from the fact that we assumed our game is $\lambda$-large (Definition \ref{def:sensitive}), and by linearity of expectation.  A similar argument shows for $j \neq i$:
\begin{align*}
 l_{j,2}^a \geq {l'}_{j,2}^a - \lambda.
\end{align*}

We then add noise to get $\hat{L}_{j,2}^a =  l_{j,2}^a + Z_{j,2}^a$ and $\hat{L'}_{j,2}^a =  {l'}_{j,2}^a + Z_{j,2}^a$ where the sensitivity between $l_{j,2}^a$ and ${l'}_{j,2}^a$  is $\lambda$ for $j \neq i$ and again $Z_{j,2}^a$ are Laplace random variables with parameter $b$.  We are then fixing the realizations of $\hat{L}_{j,2}^a$ and $\hat{L'}_{j,2}^a$ to be $\hat{l}_{j,2}^a$ for $j \neq i$.  Once again, the distributions $\pi_{j,3} = \cA(\pi_{j,2},\hat{l}_{j,2})$ do not change for $j\neq i$ between the two input types, but $\pi_{i,3}$ and ${\pi'}_{i,3}$ can change arbitrarily.  

Continuing via induction on the rounds $t$, we see that fixing the output $(\hat{l}_{-i,t}^a )_{a \in [k]}$ between the input types, allows for the losses $l_{j,t+1}^a$ and ${l'}_{j,t+1}^a$ to differ by as much as $\lambda$ between the two input types.  The noisy losses $\noisylosses_{-i,1}, \dots, \noisylosses_{-i,t}$ have already been computed when the mechanism reaches round $t+1$, thus the mechanism fits the definition of adaptive composition.  

Thus, we have rephrased the output $ (\noisylosses_{i',1}, \dots, \noisylosses_{i',T})$ as computing the answers to $nkT$ (adaptively chosen) queries on $(\tau_{1}, \dots, \tau_{n})$, each of which is $\lambda$-sensitive to the input $\tau_{i}$.  Thus the theorem follows from our choice of $b = \lambda \eps^{-1} \sqrt{8 nkT \log(1/\delta)}$ and Theorems~\ref{thm:advcomp} and~\ref{thm:laplaceprivacy}.
\end{proof}

\section{Proof of the Lower Bound on Error (Theorem \ref{thm:lb})}
To begin, as an overview of the proof for the reader , we provide a sketch of the proof of Theorem \ref{thm:lb}.  Let $\mathbf{d}\in \bits^n$ be an $n$-bit database and $\cQ = \set{q_1, \dots, q_{n'}}$ be a set of $n'$ subset-sum queries.  For the sketch, assume that we have an algorithm that computes exact equilibria.  We will split the $(n+n')$ players into $n$ ``data players'' and $n'$ ``query players.''  Roughly speaking, the data players will have utility functions that force them to play ``0'' or ``1'', so that their actions actually represent the database $\mathbf{d}$.  Each of the query players will represent a subset-sum query $q$, and we will try to set up their utility function in such a way that it forces them to take an action that corresponds to an approximate answer to $q(\mathbf{d})$.  In order to do this, first assume there are $n+1$ possible actions, denoted $\set{0,\frac{1}{n}, \frac{2}{n}, \dots, 1}$.  We can set up the utility function so that for each action $a$, he receives a payoff that is maximized when an $a$ fraction of the data players in $q$ are playing $1$.   That is, when playing action $a$, his payoff is maximized when $q(\mathbf{d}) = a$.  Conversely, he will play the action $a$ that is closest to the true answer $q(\mathbf{d})$.  Thus, we can read off the answer to $q$ from his equilibrium action.  Using each of the $n'$ query players to answer a different query, we can compute answers to $n'$ queries.  Finally, notice that joint differential privacy says that all of the actions of the query players will satisfy (standard) differential privacy with respect to the inputs of the data players, thus the answers we read off will be differentially private (in the standard sense) with respect to the database.

This sketch does not address two important issues.  The first is that we do not assume that the algorithm computes an exact equilibrium, only that it computes an approximate equilibrium.  This relaxation means that the data players do not have to play the correct bit with probability $1$, and the query players do not have to choose the answer that exactly maximizes their utility.  Below, we show that the error in the answers we read off is only a small factor larger than the error in the equilibrium computed.

The second is that we do not want to assume that the (query) players have $n+1$ available actions.  Instead, we use $\log n$ players per query, and use each to compute roughly one bit of the answer, rather than the whole answer.  However, if the query players' utility actually depends on a specific bit of the answer, then a single data player changing his action might result in a large change in utility. Below, we show how to compute bits of the answer using $1/n$-sensitive utility functions.

\begin{comment}
\begin{remark}
We remark that we used $O(n)$ \emph{linear} queries in proving our lower bound, for which a lower bound of $\Omega(1/\sqrt{n})$ is known for $(\epsilon, \delta)$-differentially private algorithms.  Thus, our $\Omega(1/\sqrt{n \log n})$ lower bound also applies to games with \emph{linear} utility functions.  However, stronger lower bounds of $\Omega(1)$ are known for answering $O(n)$ low sensitivity \emph{nonlinear} queries on a binary valued database \cite{De12} while preserving $(\epsilon,0)$-differential privacy. We could equally well use the queries from the lower bound argument of \cite{De12} in our construction, to show that no $(\epsilon,0)$-jointly differentially private algorithm can compute an $\eta$-approximate CCE to an $n$-player, $2$-action, sensitivity $1/n$ game for any $\eta < c$, where $c$ is some fixed universal constant. This proves a strong separation between $(\epsilon,\delta)$-private equilibrium computation for $\delta > 0$, and $(\epsilon,0)$-private equilibrium computation.  In particular, with $(\epsilon,0)$-privacy, it is not possible to compute an approximate equilibrium where the approximation factor tends to $0$ with the number of players, and therefore not possible to get the ``strategyproofness in the large'' results that we are able to obtain when $\delta > 0$.
\end{remark}
\end{comment}

Given a database $D \in \{0,1\}^n$, $D = (d_1, \dots, d_n)$ and $n'$ queries $\cQ = \set{q_1, \dots, q_{n'}}$, we will construct the following $N$-player $2$-action game where $N  = n + n' \log n$.  We denote the set of actions for each player by $A = \bits$.  We also use $\set{(j, h)}_{j \in [n'], h \in [\log n]}$ to denote the $n' \log n$ players $\set{n + 1, \dots, n + n' \log n}$.  For intuition, think of player $(j,h)$ as computing the $h$-th bit of $q_j(D)$.  

We fix the type profile $\tau \in \mathcal{T}^n$.  Each player $i \in [n]$ has the utility function
  \begin{equation*}
    u(\tau_i,\ba) = \begin{cases}
     1   &\qquad\text{if $a_i = d_i$} \\
     0   &\qquad\text{otherwise}
    \end{cases}
  \end{equation*}
 That is, player $i$ receives utility $1$ if they play the action matching the $i$-th entry in $D$, and utility $0$ otherwise.  Clearly, these are $0$-sensitive utility functions because it is independent of each player's type. 

 The specification of the utility functions for the query players $(j,h)$ is somewhat more complicated.  First, we define the functions $f_{h}, g_{h} \from [0,1] \to [0,1]$ as
 \begin{align*}
     &f_{h}(x) =
     1 - \min_{r \in \set{0,\dots,2^{h-1} - 1}} \left| x - (2^{-(h+1)} + r 2^{-(h-1)}) \right| \\
     &g_{h}(x) =
     1 - \min_{r \in \set{0,\dots,2^{h-1} - 1}} \left| x - (2^{-h} + 2^{-(h+1)} + r 2^{-(h-1)}) \right|
 \end{align*}
 Each $(j,h)$ player will have the utility function based on her type $\tau_{(j,h)}$
 \begin{align*}
     &u(\tau_{(j,h)},(0,a_{-(j,h)})) =
     f_{h}(q_j(a_1, \dots, a_n)) \\
     &u(\tau_{(j,h)},(1,a_{-(j,h)})) =
    g_{h}(q_j(a_1, \dots, a_n))
 \end{align*}
Since $q(a_1, \dots, a_n)$ is defined to be $1/n$-sensitive in the actions $a_1, \dots, a_n$, and $f_{h}, g_{h}$ are $1$-Lipschitz, $u(\tau_{(j,h)},\cdot)$ is also $1/n$-sensitive.  This means that the described game is $1/n$ large.  

 %  For each utility function and action, there is a set of ``peaks'', spaced at intervals of $2^{-(h-1)}$ and the utility is one minus the distance between $q(a_1, \dots, a_{n})$ and the closest peak.  Depending on whether $0$ or $1$ is played, the peaks will be offset by a distance of $2^{-h}$. \jnote{Embarrassing as it would be.  This could all be extremely clear from a picture.}  For example, for $h = 3$, the peaks for action $0$ will be $\set{\frac{1}{16}, \frac{5}{16}, \frac{9}{16}, \frac{13}{16}}$ and for action $1$ will be $\set{\frac{3}{16}, \frac{7}{16}, \frac{11}{16}, \frac{15}{16}}$.  The payoffs for the two actions also have a set of ``crossings'', spaced at intervals of $2^{-h}$.  For example, for $h=3$, the payoffs cross at the points $\set{\frac{1}{8}, \frac{2}{8}, \dots, \frac{7}{8}, 1}$.  Ignoring approximation, player $(j,h)$ would prefer to play $a_{(j,h)} = 0$ when $q(a_1, \dots, a_n) \in [0,\frac{1}{8}] \cup [\frac{2}{8},\frac{3}{8}] \cup [\frac{4}{8}, \frac{5}{8}] \cup [\frac{6}{8}, \frac{7}{8}]$ and $a_{(j,h)} = 1$ otherwise.

Also notice that since $\cQ$ is part of the definition of the game, we can simply define the types space to be all those we have given to the players.  For the data players we only used $2$ distinct types, and each of the $n' \log n$ query players may have a distinct type.  Thus we only need the set $\mathcal{T}$ to have $n' \log n + 2$ types in order to implement the reduction.

 Now we can analyze the structure of $\alpha$-approximate equilibrium in this game, and show how, given any equilibrium set of strategies for the query players, we can compute a set of $O(\alpha)$-approximate answers to the set of queries $\cQ$.

We start by claiming that in any $\alpha$-approximate CCE, every data player plays the action $d_i$ in most rounds.  Specifically,
 \begin{claim}\label{clm:accuratedataplayers}
 Let $\pi$ be any distribution over $A^N$ that constitutes an $\alpha$-approximate CCE of the game described above.  Then for every data player $i$,
\begin{align*}
\Prob{\pi}{a_i \neq d_i} \leq \alpha.
\end{align*}
 \end{claim}
\begin{proof}
\begin{align}
\Prob{\pi}{a_i \neq d_i}
={} &1 - \Ex{\pi}{u(\tau_i,(a_i, a_{-i}))} \notag \\
\leq{} &1 - \left(\Ex{\pi}{u(\tau_i,(d_i, a_{-i}))} - \alpha \right) &&(\text{Definition of $\alpha$-approximate CCE}) \notag \\
={} &1 - \left( 1 - \alpha \right) = \alpha &&(\text{Definition of $u(\tau_i,\cdot)$}) %\qedhere
\end{align}
 \end{proof}

The next claim asserts that if we view the actions of the data players, $a_1, \dots, a_n$, as a database, then $q(a_1, \dots, a_n)$ is close to $q(d_1, \dots, d_n)$ on average.
\begin{claim} \label{clm:accuratequeries}
Let $\pi$ be any distribution over $A^N$ that constitutes an $\alpha$-approximate CCE of the game described above.  Let $S \subseteq [n]$ be any subset-sum query.  Then
$$
\Ex{\pi}{\left| q(d_1, \dots, d_n) - q(a_1, \dots, a_n) \right|} \leq \alpha.
$$
\end{claim}
\begin{proof}
\begin{align}
&\Ex{\pi}{\left| q(d_1, \dots, d_n) - q(a_1, \dots, a_n) \right|}
={} \Ex{\pi}{ \frac{1}{n} \sum_{i \in S} (d_i - a_i)} &&\notag \\
&\leq{} \frac{1}{n} \sum_{i \in S} \Ex{ \pi}{ \left| d_i - a_i \right| }
={} \frac{1}{n} \sum_{i \in S} \Prob{ \pi}{a_i \neq d_i} \notag &&\\
&\leq{} \frac{1}{n} \sum_{i \in S} \alpha \quad\leq{} \alpha &&(\text{Claim~\ref{clm:accuratedataplayers}, $S\subseteq [n]$}) %\qedhere
\end{align}
\end{proof}

We now prove a useful lemma that relates the expected utility of an action (under any distribution) to the expected difference between $q_j(a_1, \dots, a_n)$ and $q_j(D)$.
\begin{claim} \label{lem:utilityvserror}
Let $\mu$ be any distribution over $A^N$.  Then for any query player $(j,h)$,
\begin{align*}
&\left| \Ex{ \mu}{u(\tau_{(j,h)},(0,\ba_{-(j,h)}))} - f_{h}(q_j(D)) \right| \leq \Ex{ \mu}{\left| q_j(a_1, \dots, a_n) - q_j(D) \right|} \textrm{, and} \\
&\left| \Ex{ \mu}{u(\tau_{(j,h)},(1,\ba_{-(j,h)}))} - g_{h}(q_j(D)) \right| \leq \Ex{ \mu}{\left| q_j(a_1, \dots, a_n) - q_j(D) \right|}.
\end{align*}
\end{claim}
\begin{proof}
We prove the first assertion, the proof of the second is identical.
\begin{align}
&\left| \Ex{\mu}{u(\tau_{(j,h)},(0,\ba_{-(j,h)}))} - f_{h}(q_{j}(D)) \right|&& \notag \\
&\qquad ={} \left| \Ex{ \mu}{f_h(q_j(a_1, \dots, a_n)) -f_{h}(q_{j}(D))} \right|&& \notag \\
&\qquad \leq{} \Ex{\mu}{\left| q_j(a_1, \dots, a_n) - q_j(D) \right|} &&(\text{$f_{h}$ is $1$-Lipschitz}) %\qedhere
\end{align}
\end{proof}

The next claim, which establishes a lower bound on the expected utility player $(j,h)$ will obtain for playing a fixed action, is an easy consequence of Claims \ref{clm:accuratequeries} and \ref{lem:utilityvserror}.
\begin{claim} \label{clm:utilityvsanswer}
Let $\pi$ be any distribution over $A^N$ that constitutes an $\alpha$-approximate CCE of the game described above.  Then for every query player $(j,h)$,
\begin{align*}
&\left| \Ex{\pi}{u(\tau_{(j,h)},(0,\ba_{-(j,h)}))} - f_{h}(q_{j}(D)) \right| \leq \alpha \textrm{, and} \\
&\left| \Ex{\pi}{u(\tau_{(j,h)},(1,\ba_{-(j,h)}))} - g_{h}(q_{j}(D)) \right| \leq \alpha.
\end{align*}
\end{claim}

Now we state a simple fact about the functions $f_{h}$ and $g_{h}$.  Informally, this asserts that we can find alternating intervals of width nearly $2^{-h}$, that nearly partition $[0,1]$, in which $f_h(x)$ is significantly larger than $g_h(x)$ or vice versa.
\begin{obs}
Let $\beta \leq 2^{-(h+1)}$.  If
\begin{align*}
x \in \bigcup_{r \in \set{0,1,\dots,2^{h-1} - 1}} \left( r 2^{-h} + \beta, (r+1) 2^{-h} - \beta \right)
\end{align*}
then $f_h(x) > g_h(x) + \beta$.  We denote this region $F_{h, \beta}$.  Similarly, if
\begin{align*}
x \in \bigcup_{r \in \set{0,1,\dots,2^{h-1} - 1}} \left( (r+1) 2^{-h} + \beta, (r+2) 2^{-h} - \beta \right)
\end{align*}
then $g_{h}(x) > f_h(x) + \beta$.  We denote this region $G_{h, \beta}$
\end{obs}

\bigskip
For example, when $h = 3$, $F_{3, \beta} = [0,\frac{1}{8} - \beta] \cup [\frac{2}{8} + \beta, \frac{3}{8} - \beta] \cup [\frac{4}{8} + \beta, \frac{5}{8} - \beta] \cup [\frac{6}{8} + \beta, \frac{7}{8} - \beta]$.  By combining this fact, with Claim~\ref{clm:utilityvsanswer}, we can show that if $q_j(D)$ falls in the region $F_{h, \alpha}$, then in an $\alpha$-approximate CCE, player $(j,h)$ must be playing action $0$ `often'.
\begin{claim} \label{clm:queryplayers}
Let $\pi$ be any distribution over $A^N$ that constitutes an $\alpha$-approximate CCE of the game described above.  Let $j \in [n']$ and $2^{-h} \geq 10\alpha$.  Then, if $q_j(D) \in F_{h,9\alpha}$,
$
\Prob{ \pi}{a_i = 0} \geq 2/3.
$
Similarly, if $q_j(D) \in G_{h,9\alpha}$, then
$
\Prob{\pi}{a_i = 1} \geq 2/3.
$
\end{claim}
\begin{proof}
We prove the first assertion.  The proof of the second is identical.  If player $(j,h)$ plays the fixed action $0$, then, by Claim~\ref{clm:utilityvsanswer},
$$
\Ex{ \pi}{u(\tau_{(j,h)},(0,\ba_{-(j,h)}))} \geq f_h(q_j(D)) - \alpha.
$$
Thus, if $\pi$ is an $\alpha$-approximate CCE, player $(j,h)$ must receive at least $f_h(q_j(D)) - 2\alpha$ under $\pi$.  Assume towards a contradiction that $\prob{a_{(j,h)} = 0} < 2/3$.  We can bound player $(j,h)$'s expected utility as follows:
\begin{align}
&\Ex{\ba \sim \pi}{u(\tau_{(j,h)},\ba)} \notag \\
& \qquad ={} \prob{a_{(j,h)} = 0} \Ex{ \pi | a_{(j,h)} = 0 }{u(\tau_{(j,h)},(0,\ba_{-(j,h)}))} \notag \\
&\qquad \quad + \prob{a_{(j,h)} = 1} \Ex{ \pi | a_{(j,h)} = 1}{u(\tau_{(j,h)},(1,\ba_{-(j,h)}))} \notag \\
&\qquad \leq{} \prob{a_{(j,h)} = 0} \left( f_h(q_j(D)) + \Ex{ \pi | a_{(j,h)} = 0}{\left| q_j(a_1, \dots, a_n) - q_j(D) \right| } \right) \notag \\
& \qquad \quad + \prob{a_{(j,h)} = 1} \left( g_h(q_j(D)) + \Ex{ \pi | a_{(j,h)} = 1}{\left| q_j(a_1, \dots, a_n) - q_j(D) \right| } \right) \label{eq:distortion} \\
& \qquad ={} f_h(q_j(D)) + \Ex{\pi}{\left| q_j(a_1, \dots, a_n) - q_j(D) \right|} - \prob{a_{(j,h)} = 1}\left( f_h(q_j(D)) - g_h(q_j(D)) \right) \notag \\
& \qquad \leq{} f_h(q_j(D)) + \alpha - 9 \alpha \prob{a_{(j,h)} = 1} \label{eq:adv} \\
&\qquad <{} f_h(q_j(D)) - 2 \alpha \label{eq:final}
\end{align}
Line~\eqref{eq:distortion} follows from the Claim~\ref{lem:utilityvserror} (applied to the distributions $\pi \mid a_{(j,h)} = 0$ and $\pi \mid a_{(j,h)} = 1$). Line~\eqref{eq:adv} follows from Claim~\ref{clm:accuratequeries} (applied to the expectation in the second term) and the fact that $q_j(D) \in F_{h, 9\alpha}$ (applied to the difference in the final term).  Line~\eqref{eq:final} follows from the assumption that $\prob{a_{(j,h)} = 0} < 2/3$.  Thus we have established a contradiction to the fact that $\pi$ is an $\alpha$-approximate CCE.
\end{proof}

Given the previous claim, the rest of the proof is fairly straightforward.  For each query $j$, we will start at $h=1$ and consider two cases: If player $(j,1)$ plays $0$ and $1$ with roughly equal probability, then we must have that $q_j(D) \not\in F_{1, 9\alpha} \cup G_{1, 9\alpha}$.  It is easy to see that this will confine $q_j(D)$ to an interval of width $18\alpha$, and we can stop.  If player $(j,1)$ does play one action, say $0$, a significant majority of the time, then we will know that $q_j(D) \in F_{1, 9\alpha}$, which is an interval of width $1/2 - 9\alpha$.  However, now we can consider $h=2$ and repeat the case analysis: Either $(j,2)$ does not significantly favor one action, in which case we know that $q_j(D) \not\in F_{2, 9\alpha} \cup G_{2, 9\alpha}$, which confines $q_j(D)$ to the union of two intervals, each of width $18\alpha$.  However, only one of these intervals will be contained in $F_{1, 9\alpha}$, which we know contains $q_j(D)$.  Thus, if we are in this case, we have learned $q_j(D)$ to within $18\alpha$ and can stop. Otherwise, if player $(j,2)$ plays, say, $0$ a significant majority of the time, then we know that $q_j(D) \in F_{1, 9\alpha} \cap F_{2, 9\alpha}$, which is an interval of width $1/4 - 9\alpha$.  It is not too difficult to see that we can repeat this process as long as $2^{-h} \geq 18\alpha$, and we will terminate with an interval of width at most $36\alpha$ that contains $q_j(D)$.

\section{Computing Equilibria in Games with Many Actions}\label{sec:many_actions}

In this section we construct and algorithm for computing equilibria in games with many actions but bounded type spaces.  First, we will formally state the privacy and accuracy guarantees of the Median Mechanism (see \cite{RR10,HR10}).

\begin{theorem}[Median Mechanism For General Queries] \label{thm:MM}
Consider the following $R$-round experiment between a mechanism $\MM$, who holds a tuple $\tau_1, \dots, \tau_n \in \mathcal{T}$, and an adaptive querier $\cB$.  For every round $r = 1,2,\dots, R$:
\begin{enumerate}
\item $\cB(Q_{1}, a_{1}, \dots, Q_{r-1}, a_{r-1}) = Q_{r}$, where $Q_{r}$ is a $\lambda$-sensitive query.
\item $a_{r} \sim \MM(\tau_1, \dots, \tau_{n}; Q_{r})$.
\end{enumerate}
For every $\eps, \delta, \lambda, \beta \in (0,1], n, R, |\mathcal{T}| \in \N$, there is a mechanism $\MM$ such that for every $\cB$
\begin{enumerate}
\item The transcript $(Q_{1}, a_{1}, \dots, Q_{R}, a_{R})$ satisfies $(\eps, \delta)$-differential privacy.
\item With probability $1-\beta$ (over the randomizations of $\MM$), $| a_{r} - Q_{r}(\tau_1, \dots, \tau_{n}) | \leq \alpha_{\MM}$ for every $r = 1,2,\dots,R$ and for
$$
\alpha_{\MM} = 16 \eps^{-1} \lambda \sqrt{n \log |\mathcal{T}|} \log(2R/\beta) \log(4/\delta).
$$
\end{enumerate}
\end{theorem}

\subsection{Noisy No-Regret via the Median Mechanism}
We now define our algorithm for computing equilibria in games with exponentially many actions.

To keep notation straight, we will use $\tau= (\tau_{1}, \dots, \tau_{n})$ to denote the types of all the $n$ players. %and $v \in \univ$ to denote a utility function considered within the mechanism. 
%Let $U = |\mathcal{T}|$, the size of type space.

First we sketch some intuition for how the mechanism works.  In particular, why we cannot simply substitute the Median Mechanism for the Laplace mechanism and get a better error bound.  Recall the queries we used in analyzing the Laplace-based algorithm $\nrl$ in the Proof of Theorem \ref{thm:privatecnrl}.  We were able to argue that fixing $\tau_{-i}$ and the previous noisy losses, the query was low-sensitivity as a function of its input $\tau_{i}$.  This argument relied on the fact that we were effectively running \emph{independent} copies of the Laplace mechanism, which guarantees that the answers given to each query do not explicitly depend on the previous queries that were asked (although the queries themselves may be correlated).  However, in the mechanism we are about to define, the queries are all answered using a \emph{single} instantiation of the Median Mechanism.  The Median Mechanism correlates its answers across queries, and thus the answers to one query may depend on the previous queries that were made.  This fact will be problematic, because the description of the queries contains the types $\tau_{-i}$.  Thus, the queries we made to construct the output for players other than $i$ will actually contain information about $\tau_{-i}$, and we cannot guarantee that this information does not leak into the answers given to other sets of players.

We address this problem by asking a larger set of queries whose description does not depend on any particular player's type.  We will make the set of queries large enough that they will actually contain every query that we might possibly have asked in $\nrl$, and each player can select from the larger set of answers only those which she needs to compute her losses.  Since the queries do not depend on any type, we do not have to worry about leaking the description of the queries.

In order to specify the mechanism it will be easier to define the following family of queries first. Let  $i$ be any player, $j$ any action, $t$ any round of the algorithm, and $\hat\tau$ any type in $\mathcal{T}$.  The queries will be specified by these parameters and a sequence $\Lambda_{1}, \dots, \Lambda_{t-1}$ where $\Lambda_{t'} \in \R^{n \times k \times |\cT|}$ for every $1 \leq t' \leq t-1$.
Intuitively, the query is given a description of the ``state'' of the mechanism in all previous rounds.  Each state variable $\Lambda_{t}$ encodes the losses that would be experienced by every possible player $i$ and every action $j$ and every type $\hat\tau$, \emph{given that the previous $t-1$ rounds of the mechanism were played using the real vector of types}.  We will think of the variables $\Lambda_{1}, \dots, \Lambda_{t-1}$ as having been previously sanitized, and thus we do not have to worry about the fact that these state variables encode information about the actual types of the players.
\begin{algorithm}
 $\cQ^{j}_{i,t,\hat\tau}(\tau_{1}, \dots, \tau_{n} \mid \Lambda_{1}, \dots, \Lambda_{t-1})$
\begin{algorithmic}
\State Using $\tau_{1}, \dots, \tau_{n} \mid \Lambda_{1}, \dots, \Lambda_{t-1}$, compute $\loss^j_{i,t,\hat\tau} = 1 - \Ex{\state_{-i,t}}{u(\tau_i,(j,a_{-i}))}$.  
\State This computation can be done in the following steps:
\begin{enumerate}
\item For every $i' \neq i$, use $\Lambda^{j}_{i',1, \tau_{i'}}, \dots, \Lambda^{j}_{i',t-1, \tau_{i'}}$, $\nralg$, and $\tau_{i'}$ to compute $\state_{i',1}, \dots, \state_{i',t-1}$.
\item Using $\state_{-i,t-1}$, compute $\loss^j_{i,t,\hat\tau}$.
\end{enumerate}
\end{algorithmic}
\end{algorithm}

Observe that $Q^{j}_{i,t,\hat\tau}$ is $\lambda$-sensitive for every player $i$, step $t$, action $j$, and type $\hat\tau$.  To see why, consider what happens when a specific player $i'$ switches her input from $\tau_{i'}$ to $\tau'_{i'}$.  In that case that $i = i'$, this has no effect on the query answer, because player $i$'s utility is never used in computing $Q^{j}_{i,t,\hat\tau}$.  In the case that $i' \neq i$ then the type of player $i'$ can (potentially) affect the computation of $\state_{i',t-1}$, and can (potentially) change it to an arbitrary state $\overline{\state}_{i', t-1}$.  But then $\lambda$-sensitivity follows from the $\lambda$-sensitivity of $u(\tau_{i},\cdot)$, the definition of $\itjvloss$, and linearity of expectation.  Notice that $\tau_{i'}$ does not, however, affect the state of any other players, who will use the losses $\Lambda_{1}, \dots, \Lambda_{t-1}$ to generate their states, not the actual states of the other players.

Now that we have this family of queries in places, we can describe the algorithm.  Our mechanism uses two steps.  At a high level, there is an inner mechanism, $\nrms$, that will use the Median Mechanism to answer each query $Q^{j}_{i,t,\hat\tau}\left(\cdot \mid  \widehat{\Lambda}_{1}, \dots, \widehat{\Lambda}_{t-1}\right)$, and will output a set of noisy losses $\noisy{\Lambda}_{1}, \dots, \noisy{\Lambda}_{T}$.  The properties of the Median Mechanism will guarantee that these losses satisfy $(\eps, \delta)$-differential privacy (in the standard sense of Definition~\ref{def:standardprivacy}).

There is also an outer mechanism that takes these losses and, for each player, uses the losses corresponding to her type to run a no-regret algorithm. This is $\nrm$ which takes the sequence $\noisy{\Lambda}_{1}, \dots, \noisy{\Lambda}_{T}$ and using the type $\tau_{i}$ will compute the equilibrium strategy for player $i$. Since each player's output can be determined only from her own type and a set of losses that is $(\eps, \delta)$-differentially private with respect to every type, the entire mechanism will satisfy $(\eps, \delta)$-joint differential privacy.  This follows from the Billboard Lemma (Lemma \ref{lem:billboard}).

\begin{algorithm}
\caption{Inner Mechanism}\label{NRMedian_shared}
\begin{algorithmic}[0]
%\INPUT : Utility profile $(u_{1}, \dots u_{n})$
%\OUTPUT : Noisy losses $(\noisy{\Lambda}_{1}, \dots, \noisy{\Lambda}_{T})$
\Procedure {$\nrms^{\nralg}$}{$\tau_{1}, \dots \tau_{n}$}
\State Parameters $\eps, \delta, \lambda \in (0,1]$ and $n, k, T \in [n]$ 
\For{$t = 1, 2, \dots, T$}
\State Let: $\noisyitjvloss = \MM\left(\tau_{1}, \dots, \tau_{n}; Q^{j}_{i,t,\hat\tau}( \cdot \mid \widehat{\Lambda}_{1}, \dots, \widehat{\Lambda}_{t-1})\right)$ for every $i, j, \hat\tau $.
\State Let: $\widehat{\Lambda}^j(i, t, \hat\tau) = \noisyitjvloss$ for every $i , j , \hat\tau $.
\EndFor
\State \textbf{return} $(\noisy{\Lambda}_{1}, \dots, \noisy{\Lambda}_{T})$. 
\EndProcedure
\end{algorithmic}
\end{algorithm}
\begin{algorithm}
\caption{Outer Mechanism}\label{NRMedian} 
\begin{algorithmic}
\Procedure {$\nrm^{\nralg}$} {$\tau_{1}, \dots \tau_{n}$}
\State Parameters: $\eps, \delta, \Delta \in (0,1], n, k, T \in [n]$ 
\State Let: $(\noisy{\Lambda}_{1}, \dots, \noisy{\Lambda}_{T}) =\nrms^{\nralg}(\tau_{1}, \dots, \tau_{n})$.
\For{$i = 1, \dots, n$}
\State Let: $\state_{i,1}$ be the uniform distribution over $\actionset$.
	\For{ $t = 1, \dots, T$}
		\State Let: $\state_{i,t} = \nralg\left(\state_{i,t-1}, \noisy{\Lambda}_{i,t-1,\tau_{i}}\right)$
	\EndFor
\State \textbf{return} to player $i$: $(\state_{i,1}, \dots, \state_{i,T})$. 
\EndFor
\EndProcedure
\end{algorithmic}
\end{algorithm}

\begin{theorem}[Privacy of $\nrm$]\label{thm:privCCEMedian}
The algorithm $\nrm$ satisfies $(\eps, \delta)$-joint differential privacy.
\end{theorem}
\begin{proof}
We let $\M$ be $\nrms$ and write $o = \M(\tau)$.  Observe that the output of $\nrm$ can be written as $\M'(\tau) = (f_1(\tau_1,o), \dots, f_n(\tau_n,o))$ where $f_i$ depends only on $\tau_i$ for every player $i$ where $f_i$ is the $i$-th iteration of the main loop in $\nrm$.  The privacy of the Median Mechanism (Theorem~\ref{thm:MM}) directly implies that $\M$ is $(\eps, \delta)$-differentially private (in the standard sense).

%%% ADDED %%%%%
We use the Billboard Lemma to then say that $\M'(\tau) = (f_1(\tau_1,\M(\tau)), \cdots, f_n(\tau_n,\M(\tau)))$ is $(\epsilon,\delta)$ joint differentially private.  
%%%%%%%%%%%%%
%
%Consider a player $i$ and two profiles $\tau,\tau'$ that differ only in the input of player $i$, and consider the output $(f_{-i}(g(\tau)))$.  Let $S \subseteq \mathrm{Range}(f_{-i})$ and let $R(\tau) = \set{o \in \mathrm{Range}(g) \mid f^{-i}(o) \in S}$.  Notice that $f$ is deterministic, so $R$ is well-defined.  Also notice that $R$ depends only on $S$ and $\tau_{-i}$ (in particular, not on  $\tau_i$).  Then we have
%\begin{align*}
%\Prob{h(\tau)}{h^{-i}(\tau) \in S}
%&={} \Prob{g(\tau)}{g(\tau) \in R(\tau) = R(\tau')} \\
%&\leq{} e^{\eps} \Prob{g(\tau')}{g(\tau') \in R(\tau) = R(\tau')} + \delta \\
%&\leq{} e^{\eps} \Prob{h(\tau)}{h^{-i}(\tau') \in S} + \delta
%\end{align*}
%where the first inequality follows from the (standard) $(\eps, \delta)$-differential privacy of $g$.  Thus, {\bf\textsc{NRMedian}} satisfies $(\eps, \delta)$-joint differential privacy.
\end{proof}

\subsection{Computing Approximate Equilibria}

\begin{theorem}[Computing CE]\label{thm:accCEmedian}
Let $\nralg$ be $\swapalg.$ Fix the environment, i.e the number of players $n$, the number of actions $k$, number of possible types $|\mathcal{T}|$, largeness of the game $\lambda$ and desired privacy $(\epsilon, \delta)$.
Suppose $\beta$ and $T$ are such that:
\begin{equation} \label{eqn:paramassCEmedian}
16 \eps^{-1} \lambda \sqrt{n \log |\cT|} \log(2 |\cT|nkT / \beta) \log(4/\delta) \leq \tfrac{1}{6}
\end{equation}
Then with probability at least $1-\beta$ the algorithm $\nrm^{\fixedalg}$ returns an $\alpha$-approximate CE for:
%
%\footnote{Here, $\tilde{O}$ hides lower order $\poly( \log n, \log\log k, \log T, \log\log \U \log(1/\lambda), \log(1/\eps), \log\log(1/\beta), \log\log(1/\delta))$ terms.}
\begin{equation*}
\alpha = \tilde{O}\left(\frac{\lambda \sqrt{n} \log^{3/2} (|\cT|) \log(k/\beta) \log(1/\delta)}{\eps} \right).
\end{equation*}
%\begin{equation*}
%\alpha = \sqrt{\frac{2 \log K}{T}} + \frac{16 \Delta \sqrt{N \log \U} \log(2NKT\U / \beta) \log(4/\delta)}{\eps}
%\end{equation*}
\end{theorem}

Again, considering `low sensitivity' games where $\lambda$ is $O(1/n)$, the theorem says that fixing the desired degree of privacy, we can compute an $\alpha$-approximate equilibrium for $\alpha = \tilde{O}\left(\frac{\log^{\frac{3}{2}}( |\cT|) \log k}{\sqrt n} \right)$. The tradeoff to the old results is in dependence on the number of actions. The results in the previous section had a $\sqrt{k}$ dependence on the number of actions $k$. This would have no bite if $k$ grew even linearly in $n$. We show that positive results still exist if the number of possible private types is is bounded - the dependence on the number of actions and the number of types is now logarithmic. However this comes with two costs. First, we can only consider situations where the number of types any player could have is bounded, and grows sub-exponentially in $n$. Second, we lose computational tractability-- the running time of the Median Mechanism is exponential in the number of players in the game.

\begin{proof}[Proof of Theorem \ref{thm:accCEmedian}]
By the accuracy guarantees of the Median Mechanism:
\begin{equation*}
\Prob{\MM}{\exists i, t, j, \hat\tau \textrm{ s.t. } \left|\noisyitjvloss - \itjvloss\right| > \alpha_{\MM}} \leq \beta
\end{equation*}
where
\begin{equation*}
\acc_{\MM} =  16\lambda \eps^{-1} \sqrt{n \log |\cT|} \log(2|\cT| nkT/ \beta) \log(4/\delta)
\end{equation*}
By (\ref{eqn:paramassCEmedian}), $\acc_{\MM} \leq 1/6$. Therefore,
\begin{equation*}
\Prob{\MM}{\exists i,j,t,\hat\tau \textrm{ s.t. } |\noisyitjvloss - \itjvloss| > \tfrac{1}{6}} \leq \beta
\end{equation*}
Applying Theorem~\ref{thm:lowregretbounded} and substituting $\alpha_{\MM}$, we obtain:
\begin{equation*}
\Prob{\Noises}{\exists i \textrm{ s.t. } \regret(\state_{i,1}, \dots, \state_{i,T}, \Losses, \swapmods) > \sqrt{\frac{2k \log k}{T}} + 2\acc_{\MM}} \leq \beta
\end{equation*}
Now we can choose $\sqrt{T} = k(\lambda \sqrt{n})^{-1}$ to conclude the proof.
\end{proof}

\section{Additional Related Work}
\label{app:related}
The most well studied problem is that of accurately answering numeric-valued queries on a data set. A basic result of \cite{DMNS06} is that any low sensitivity query (i.e. the addition or removal of a single entry can change the value of the query by at most $1$) can be answered efficiently and ($\epsilon$-differential) privately while introducing only $O(1/\epsilon)$ error. Another fundamental result of \cite{DKMMN06,DRV10} is that differential privacy composes gracefully. Any algorithm composed of $T$ subroutines, each of which are $O(\epsilon)$-differentially private, is itself $\sqrt{T}\epsilon$-differentially private. Combined, these give an efficient algorithm for privately answering any $T$ low sensitivity queries with $O(\sqrt{T})$ effort, a result which we make use of.

Using computationally inefficient algorithms, it is possible to privately answer queries much more accurately \cite{BLR08, DRV10, RR10, HR10, GHRU11,GRU12}. Combining the results of the latter two yields an algorithm which can privately answer arbitrary low sensitivity queries as they arrive, with error that scales only logarithmically in the number of queries. We use this when we consider games with large action spaces.

Our lower bounds for privately computing equilibria use recent information theoretic lower bounds on the accuracy queries can be answered while preserving differential privacy \cite{DN03,DMT07,DY08, De12}. Namely, we construct games whose equilibria encode answers to large numbers of queries on a database.

Variants of differential privacy related to joint differential privacy have been considered in the setting of query release, specifically for analyst privacy \cite{DNV12}. Specifically, the definition of one-analyst-to-many-analyst privacy used by \cite{HRU13} can be seen as an instantiation of joint differential privacy.

\end{document}